\documentclass[aps,pra,showpacs,notitlepage,twocolumn,superscriptaddress,nofootinbib,preprintnumbers]{revtex4-1}
\usepackage{bbm}
\usepackage{mathrsfs}
\usepackage{subcaption}
\usepackage{epsfig}

\usepackage{soul,color}
\usepackage{graphicx}
\usepackage{amsfonts}
\usepackage{amsthm}
\usepackage[figuresright]{rotating}
\usepackage{amssymb}
\usepackage{amsmath}
\usepackage{dcolumn}
\usepackage{physics}
\usepackage{float}
\usepackage{bm}
\usepackage{verbatim}
\usepackage{braket}
\usepackage[normalem]{ulem}
\usepackage[ruled,vlined,linesnumbered]{algorithm2e}
\usepackage{setspace}
\usepackage[colorlinks,linkcolor=blue,anchorcolor=blue,citecolor=blue,urlcolor=blue]{hyperref}
\usepackage{stackengine}
\usepackage{qcircuit}
\newcommand\xrowht[2][0]{\addstackgap[.5\dimexpr#2\relax]{\vphantom{#1}}}
\newcommand{\subsubsubsection}[1]{{\textbf{#1}}}
\newtheorem{theorem}{Theorem}
\newtheorem{lemma}[theorem]{Lemma}

\usepackage{fancyhdr}
\newtheorem{innercustomdef}{Definition}
\newenvironment{customdef}[1]
  {\renewcommand\theinnercustomdef{#1}\innercustomdef}
  {\endinnercustomdef}

\begin{document}
\title{Efficient Fully-Coherent Quantum Signal Processing Algorithms for Real-Time Dynamics Simulation}
\author{John M. Martyn}\email{jmmartyn@mit.edu}\affiliation{Department of Physics, Center for Theoretical Physics, Massachusetts Institute of Technology, Cambridge, Massachusetts 02139, USA} 
\author{Yuan Liu}\affiliation{Department of Physics, Co-Design Center for Quantum Advantage, Massachusetts Institute of Technology, Cambridge, Massachusetts 02139, USA}
\author{Zachary E. Chin}\affiliation{Department of Electrical Engineering and Computer Science, Massachusetts Institute of Technology, Cambridge, Massachusetts 02139, USA} \affiliation{Department of Chemistry, Massachusetts Institute of Technology, Cambridge, Massachusetts 02139, USA}
\author{Isaac L. Chuang}\affiliation{Department of Physics, Co-Design Center for Quantum Advantage, Massachusetts Institute of Technology, Cambridge, Massachusetts 02139, USA} \affiliation{Department of Electrical Engineering and Computer Science, Massachusetts Institute of Technology, Cambridge, Massachusetts 02139, USA}

\begin{abstract}
    Simulating the unitary dynamics of a quantum system is a fundamental problem of quantum mechanics, in which quantum computers are believed to have significant advantage over their classical counterparts. One prominent such instance is the simulation of electronic dynamics, which plays an essential role in chemical reactions, non-equilibrium dynamics, and material design. These systems are time-\emph{dependent}, which requires that the corresponding simulation algorithm can be successfully concatenated with itself over different time intervals to reproduce the overall coherent quantum dynamics of the system. In this paper, we quantify such simulation algorithms by the property of being \emph{fully-coherent}: the algorithm succeeds with arbitrarily high success probability $1-\delta$, while only requiring a single copy of the initial state. We subsequently develop fully-coherent simulation algorithms based on quantum signal processing (QSP), including a novel algorithm that circumvents the use of amplitude amplification while also achieving a query complexity \emph{additive} in time $t$, $\ln(1/\delta)$, and $\ln(1/\epsilon)$ for error tolerance $\epsilon$: $\Theta\big( \|\mathcal{H}\| |t| + \ln(1/\epsilon) + \ln(1/\delta)\big)$. Furthermore, we numerically analyze these algorithms by applying them to the simulation of the spin dynamics of the Heisenberg model and the correlated electronic dynamics of an H$_2$ molecule. Since any electronic Hamiltonian can be mapped to a spin Hamiltonian, our algorithm can efficiently simulate time-dependent \textit{ab initio} electronic dynamics in the circuit model of quantum computation. Accordingly, it is also our hope that the present work serves a bridge between QSP-based quantum algorithms and chemical dynamics, stimulating a cross-fertilization between these exciting fields.
\end{abstract}

\preprint{MIT-CTP/5334}
\maketitle

\section{Introduction}
\fancypagestyle{plain}{%
	\fancyhead[R]{\fbox{to appear in journal xx}}
	\renewcommand{\headrulewidth}{0pt}
}

Quantum computation owes its inception to the fundamental problem of Hamiltonian simulation \cite{feynman1982simulating}, wherein one aims to simulate the time evolution of a quantum system under a Hamiltonian $\mathcal{H}$ for a time $t$. Such a time evolution underpins nearly all dynamical processes as microscopic phenomena are governed by quantum mechanics. Efficient algorithms for Hamiltonian simulation are thus crucial for theoretical modeling of the physical world, including analyzing reaction mechanisms of complex chemical and biological systems~\cite{reiher2017elucidating, Mcardle_2020, lambert2013quantum, cao2020quantum}, probing non-equilibrium ultrafast dynamics \cite{krausz2009attosecond,maiuri2020ultrafast,young2018roadmap}, understanding phases of strongly correlated condensed matter systems \cite{Klein_2006,hofstetter2018quantum}, and simulating quantum field theories~\cite{Preskill_2018, Jong_2021}.

Classical algorithms for simulating the time evolution of quantum systems have existed long before the inception of quantum computation. One major thrust of development comes from the need to model correlated electronic dynamics to high accuracy, where exciting developments have happened over the past few decades \cite{li2020real}. Besides the exact diagonalization (ED) method which scales exponentially with the system size, algorithms based on mean-field theories \cite{mclachlan1964time,jorgensen1975molecular,runge1984density} as well as their multi-configuration variants \cite{meyer1990multi,zanghellini2004testing,fromager2013multi} are among the most computationally efficient methods, yet not applicable to strongly-correlated electronic dynamics. More accurate classical methods for electronic dynamics include the time-dependent variants of high-level correlated electronic structure theories \cite{greenman2010implementation,rohringer2006configuration,schriber2019time,huber2011explicitly,white2018time,white2019time,shushkov2019real,cazalilla2002time,white2004real,daley2004time}. Many of these methods involve a truncation of the Hilbert space and have high-order polynomial scaling with the system size. In addition, real-time stochastic methods based on quantum Monte Carlo sampling \cite{foulkes2001quantum,makri1987monte,doll1987toward,schiro2010real,cohen2015taming,motta2018ab,church2021real} achieve a balance between computational cost and accuracy, yet their applicability is limited to short-time dynamics or small systems due to the notorious real-time phase problem on the sampling trajectories. In general, the phase problem cannot be resolved without truncating the sampling space and adding bias to the simulation results. A combination of the computationally efficient mean-field theories and more accurate high-level theories can be achieved in the embedding framework \cite{freericks2006nonequilibrium,kretchmer2018real} by treating the most important parts of the system using correlated methods and the rest by mean-field theories. 

These advancements in classical algorithms development for correlated electronic dynamics lean on the physical nature of classical computation. Namely, it is \emph{necessary} to use an exponential amount of classical resources to represent the full dynamics of the entire Hilbert space of an interacting many-electron system, unless approximations are used. \emph{Quantum algorithms}, on the other hand, hold the promise to overcome the difficulties encountered by classical algorithms by employing coherent superpositions and entanglement in quantum hardware~\cite{cao2019quantum,kivlichan2017bounding,kassal2008polynomial}. While previous efforts leveraging quantum resources mostly focus on electronic structure problems \cite{lee2022there,cao2019quantum}, in the present work we instead are primarily interested in quantum dynamics, or equivalently Hamiltonian simulation.

Among a variety of efforts to subdue Hamiltonian simulation, prominent quantum algorithms include Trotterization and product formulas~\cite{Lloyd1073, suzuki1991general, Berry_2006,childs2021theory,su2021nearly,tran2021faster}, Taylor series truncation~\cite{2015Berry,novo2017improved,zhao2021exploiting}, and quantum walks~\cite{Childs_2009, Berry_2015, berry2012black, zhang2021parallel}. In recent years however, Hamiltonian simulation algorithms have flourished with the advent of an algorithmic primitive known as {\em quantum signal processing} (QSP)~\cite{Low_2016,Low_2019}, with QSP-based simulation algorithms touting a nearly optimal query complexity with respect to the simulation time $t$ and error $\epsilon$. From a high level, QSP provides a systematic method to apply a nearly arbitrary polynomial transformation to a quantum subsystem. Exploiting such flexibility, QSP-based simulation develops a polynomial transformation of the Hamiltonian that approximates the time evolution operator. 

From an algorithmic viewpoint, a common theme among these Hamiltonian simulation algorithms, and quantum algorithms in general, is cascading together elementary subroutines to construct more powerful algorithms. Whereas from an application point of view, the underlying electronic Hamiltonians that governs the dynamics are often time-dependent due to nuclei motion or interaction with external fields \cite{weinberg2012proton,ramasesha2016real, calegari2014ultrafast,rotter2015review}. One common technique to deal with a time-dependent Hamiltonian is to Trotterize \cite{hen2021quantum,low2018hamiltonian,childs2021theory} it into many time-independent pieces and then individually simulate each piece before coherently cascading them together. However, this cascading procedure is only viable if each time evolution subroutine succeeds with probability near unity, which indicates the importance of maintaining \emph{coherence}, that is, retaining the correct wave function with arbitrarily high probability. 

This is precisely the problem we study in this paper. We focus on \emph{fully-coherent Hamiltonian simulation}, the goal of which is to accurately time evolve a state for a time $t$ to within error $\epsilon$ with arbitrarily high success probability $1-\delta$, provided only a single copy of the initial state. We also center our attention on \emph{efficient} simulation algorithms, whose query complexities are polynomials in $t$, $\ln(1/\epsilon)$, and $\ln(1/\delta)$. Rigorous definitions for these terms are provided later in the text.

\begin{table*}[htbp]
    \centering
    \begin{tabular}{c|c}
        \hline \hline \xrowht[()]{11pt}
        Simulation Algorithm & Query Complexity \  \\
        \hline \hline \xrowht[()]{16pt}
         QSP-LCU + Conventional AA (Sec.~\ref{Sec:ConventionalHamSim}) & $\Theta\Big( \ln(\frac{1}{\delta}) \left(\alpha |t| + \frac{\ln(1/\epsilon)}{\ln(e+\ln(1/\epsilon)/\alpha |t|)} \right) \Big)$ \\ 
         \hline \xrowht[()]{14pt}
         QSP-LCU + Robust Oblivious AA (Sec.~\ref{Sec:ConventionalHamSim}) & $\Theta\Big( \alpha |t| + \frac{\ln(1/\epsilon)}{\ln(e+\ln(1/\epsilon)/\alpha |t|)} \Big)$ \\ 
         \hline \xrowht[()]{14pt}
        Coherent One-Shot Simulation (Sec.~\ref{Sec:OneShotHamSim})& $\Theta \Big(\alpha |t| + \ln(\frac{1}{\epsilon}) + \ln(\frac{1}{\delta})\Big)$  \\
        \hline \hline
    \end{tabular}
    \caption{Query complexities of the three efficient fully-coherent Hamiltonian simulation algorithms discussed in this paper. These include conventional QSP (denoted ``QSP-LCU" in reference to the linear combination of unitaries (LCU) circuit used in its construction) augmented with amplification (AA); in particular, QSP-LCU + Conventional AA and QSP-LCU + Robust Oblivious AA, both described in Sec.~\ref{Sec:ConventionalHamSim}. We also include our Coherent One-Shot Simulation algorithm of Sec.~\ref{Sec:OneShotHamSim}. In these expressions, $t$ is the simulation time, $\alpha$ is an upper bound on the norm $\|\mathcal{H} \|$, $\epsilon$ is the error in the approximation of $e^{-i \mathcal{H}t}$, and $\delta$ is the probability of failure.}
    \label{tab:query-complex0}
\end{table*}

More concretely, we analyze three efficient fully-coherent Hamiltonian simulation algorithms, all of which are rooted in QSP. The first algorithm augments QSP-based simulation with conventional amplitude amplification to boost its success probability, which is the suggested remedy for the issue of post-selection in QSP algorithms~\cite{Gily_n_2019}. This appends a multiplicative factor of $\ln(1/\delta)$ to the query complexity. The second algorithm integrates QSP-based simulation with the robust oblivious amplitude amplification protocol of Ref.~\cite{Berry_2015}, which contributes only an additive factor of $\ln(1/\delta)$.

In contrast, the third algorithm we present introduces a novel QSP technique: it first compresses the spectrum of the Hamiltonian with an affine transformation, and subsequently applies QSP to it with a polynomial that approximates the time evolution operator only over the range of the compressed spectrum. By incorporating such a \textit{pre-transformation} before applying QSP, this algorithm circumvents the need for amplitude amplification and also attains a query complexity additive in $\ln(1/\delta)$ instead of multiplicative: $\Theta\big(\| \mathcal{H}\| |t| + \ln(1/\epsilon) + \ln(1/\delta) \big)$. We further demonstrate that the probability of failure is dictated by the error of the complex exponential approximation, such that $\delta = \Theta(\epsilon)$, which may be easily tuned by the choice of QSP polynomial. We dub this algorithm ``Coherent One-Shot Hamiltonian Simulation". In Table~\ref{tab:query-complex0}, we summarize the query complexities of the fully-coherent simulation algorithms discussed in this work.

For completeness, we note that these algorithms can be seen as alternatives to the single-ancilla QSP algorithm of Ref.~\cite{Low_2019}, which is a modified version of QSP that can also achieve fully-coherent simulation. However, whereas single-ancilla QSP is predicated on a set of constraints, which as we discuss later can inhibit its robustness to noise, here we use new techniques such as the pre-transformation to achieve fully-coherent simulation. In addition, our algorithms are constructed through recent generalizations of QSP known as the quantum eigenvalue transformation and the quantum singular value transformation, which extend the framework of qubitization used in single-ancilla QSP and enable a more intuitive algorithmic development. We hence focus on the aforementioned three algorithms, and briefly touch base with single-ancilla QSP later.

The rest of the paper is organized as follows. In Sec.~\ref{Sec:ConventionalHamSim}, we summarize conventional QSP-based Hamiltonian simulation as initially laid out in Ref.~\cite{Gily_n_2019}, and then integrate this algorithm with amplitude amplification (both conventional, and robust oblivious) to make it fully-coherent. In Sec.~\ref{Sec:OneShotHamSim}, we develop coherent one-shot Hamiltonian simulation, and then compare numerically the complexities of these algorithms as a function of simulation error and simulation time in Sec.~\ref{Sec:Scaling}. We thereafter use these algorithms in Sec.~\ref{Sec:Applications} to simulate the Heisenberg model with both time-independent and time-dependent external onsite fields, as well as electronic migration dynamics in a hydrogen molecule. Finally, Sec.~\ref{Sec:Discussion} concludes this work and discusses future research directions. In aggregate, we hope that this paper serves as an introduction to QSP-based quantum algorithms for the electronic dynamics community and spurs a cross-pollination between the two fields.

\section{Conventional QSP-Based Hamiltonian Simulation} \label{Sec:ConventionalHamSim}

In order to establish the foundations of this work, we first review the conventional QSP-based Hamiltonian simulation algorithm based on a linear combination of untiaries (LCU) circuit. Paralleling the original presentation in Ref.~\cite{Gily_n_2019}, we describe the procedure of this algorithm in Sec.~\ref{sec:qsp-lcu-procedure}, followed by a query complexity analysis in Sec.~\ref{Sec:Conventional_Complexity}. For notational simplicity, we shall refer this approach as the ``QSP-LCU" method. In Sec.~\ref{sec:full-coherent-aa}, we discuss two methods for constructing fully coherent versions of the QSP-LCU method via amplitude amplification, the first using conventional amplitude amplification and the second employing the robust oblivious amplitude amplification protocol of Ref.~\cite{2015Berry}.

This section, as well as the rest of this paper, will assume familiarity with the basics of QSP, the quantum eigenvalue transformation (QET), and the quantum singular value transformation (QSVT). For a brief review of these fundamental concepts, see Appendix~\ref{Sec:qsp}.

\subsection{Procedure} \label{sec:qsp-lcu-procedure}
    In the setup of this problem, we assume access to the Hamiltonian $\mathcal{H}$, of which we desire a unitary block encoding such that we may solve this problem with QSP techniques, in particular the quantum eigenvalue transformation (QET). However, such a unitary block encoding is only realizable if $\|\mathcal{H}\| \leq 1$. In general then, we instead determine an $\alpha \geq \| \mathcal{H}\|$ and construct a unitary block encoding of $\mathcal{H}/\alpha$. This requires some prior knowledge about $\mathcal{H}$, but fortunately such a block encoding can be achieved for a large class of Hamiltonians~\cite{Low_2019, Gily_n_2019}. For simplicity, we will also assume that $\mathcal{H}/\alpha$ is encoded in the $|0\rangle \langle 0|$ matrix element of the unitary, but the following results may be easily adapted to other encodings as well. 
    
    With this rescaled block encoding, one can equivalently imagine that our goal is to simulate the time evolution of a system under the rescaled Hamiltonian $\mathcal{H}/\alpha$ for an effective time $\tau = t \alpha$. This equivalence holds because the corresponding time evolution operators are identical: $e^{-i(\mathcal{H}/\alpha) (\alpha t)} = e^{-i\mathcal{H}t}$.
    
    Hamiltonian simulation may then be straightforwardly achieved with QET. Naively, one may try to employ QET with a polynomial approximation to the function $e^{-ix\tau}$ (here, we view $\tau$ as a parameter representing effective time, not a variable), but because the complex exponential does not have definite parity, this function does not satisfy the constraints on $P(x)$ discussed in Appendix~\ref{Sec:qsp}, thus rendering this approach infeasible. Following~\cite{Gily_n_2019, martyn2021grand}, one can circumvent this issue by instead applying QET twice -- once with an even polynomial approximation to $\cos(x\tau)$, and once with an odd polynomial approximation to $-i \sin(x\tau)$, both of which have definite parities. In addition, although these functions do not obey $|P(\pm 1)|=1$, they may still be implemented by looking at the $|+\rangle\langle+|$ component of the QET sequence, as discussed in Appendix~\ref{Sec:qsp_Overview}. Then, using the linear combination of unitaries (LCU) circuit illustrated in Fig.~\ref{fig:ComplexExponentialCircuit}, one can sum together the results of these two QET executions to obtain $\cos(\mathcal{H}t) - i\sin(\mathcal{H}t) = e^{-i\mathcal{H}t}$, as desired. We refer to this algorithm as conventional QSP-LCU Hamiltonian simulation.

    \begin{figure*}[htbp]
        \begin{center}
        \includegraphics[width=11.5cm]{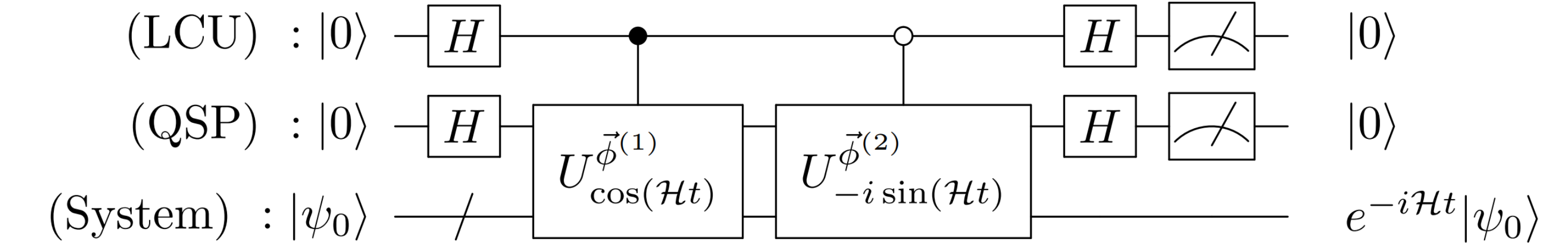}
        \end{center}
    \caption{A quantum circuit, known as a linear combination of unitaries (LCU) circuit, that can apply the time evolution operator $\cos(\mathcal{H}t) - i\sin(\mathcal{H}t) = e^{-i\mathcal{H}t}$ to $|\psi_0\rangle$ in conventional QSP-based Hamiltonian simulation. Here, $U^{\vec{\phi}^{(1)}}_{\cos(\mathcal{H}t)}$ and $U^{\vec{\phi}^{(2)}}_{-i\sin(\mathcal{H}t)}$ are unitaries that block encode $\cos(\mathcal{H}t)$ and $-i\sin(\mathcal{H}t)$, respectively, in their $|+\rangle \langle +|$ matrix elements (hence the Hadamards applied to the QSP qubit, which we have denoted by (QSP)), and may be constructed as QET sequences using phases $\vec{\phi^{(1)}}$ and $\vec{\phi^{(2)}}$, respectively, and a QSP qubit. The correct evolution of the input state is achieved only upon post-selection of both the LCU qubit (denoted (LCU)) and the QSP qubit in the $| 0\rangle$ state. This occurs with a probability close to $\frac{1}{4}$, hence requiring amplitude amplification or repetition to increase the success probability.}
    \label{fig:ComplexExponentialCircuit}
    \end{figure*}

    How likely is this method to succeed and produce the correct time evolved state? As indicated in Fig.~\ref{fig:ComplexExponentialCircuit}, this procedure only succeeds when the LCU qubit (used to achieve the linear combination of unitaries) and the QSP qubit (used to achieve the block encoding) are both measured in the state $|0\rangle$, so we are interested in the probability of accessing the $|00\rangle \langle 00|$ matrix element. Supposing that our choice of QET polynomial allows us to construct an $\epsilon$-approximation to $e^{-i\mathcal{H}t}$, we find that the $|00\rangle \langle 00|$ matrix element of the entire unitary transformation illustrated in Fig.~\ref{fig:ComplexExponentialCircuit} is an $\epsilon/2$-approximation to $\frac{1}{2} e^{-i\mathcal{H}t}$, which we denote by $\frac{1}{2} (e^{-i\mathcal{H}t} \pm \epsilon)$. Accordingly, the probability of success is $p = \big\| \frac{1}{2} (e^{-i\mathcal{H}t} \pm \epsilon) \big\|^2$, i.e. $ \frac{(1-\epsilon)^2}{4} \leq p \leq \frac{(1+\epsilon)^2}{4} $, which is close to $\frac{1}{4}$ and must be corrected to achieve arbitrarily high success probability.

    \subsection{Query Complexity}\label{Sec:Conventional_Complexity}
    Let us now look at the problem of constructing an appropriate QSP polynomial, from which the query complexity may be extracted as $\mathcal{O}(d)$ where $d$ is the degree of the polynomial. We note that in Ref.~\cite{Gily_n_2019}, Gily\'en \textit{et al.} approximate the functions $\cos(x\tau)$ and $\sin(x\tau)$ by polynomials using the Jacobi-Anger expansion:
    \begin{align} \label{eq:Jacobi-Anger1}
        \cos(x\tau) &= J_0(\tau) + 2\sum_{k=1}^\infty (-1)^k J_{2k}(\tau) T_{2k}(x) \\
        \label{eq:Jacobi-Anger2}
        \sin(x\tau) &= 2\sum_{k=0}^\infty (-1)^k J_{2k+1}(\tau) T_{2k+1}(x),
    \end{align}
    where $J_i(x)$ and $T_i(x)$ are the Bessel function and Chebyshev polynomial of order $i$, respectively. One can attain $\epsilon$-approximations to $\cos(x\tau)$ and $\sin(x\tau)$ by truncating these infinite series at a sufficiently large index $K$. The necessary truncation index $K$ may be determined by a function $r(\tau,\epsilon)$, which is defined implicitly as
        \begin{equation}
            \epsilon = \left( \frac{|\tau|}{r(\tau, \epsilon)} \right)^{r(\tau, \epsilon)} \text{ s.t. }  r(\tau, \epsilon) \in (|\tau|, \infty).
        \end{equation}
    $r(\tau, \epsilon)$ may be solved for as $r(\tau,\epsilon) = |\tau| e^{W(\ln(1/\epsilon)/|\tau|)}$, where $W(x)$ is the Lambert-$W$ function, and is proven to scale as
        \begin{equation}
            r(\tau, \epsilon) = \Theta\left(|\tau| + \frac{\ln(1/\epsilon)}{\ln\big(e+\ln(1/\epsilon)/|\tau|\big)} \right).
        \end{equation}
    Returning to the series, it is proven that truncating Eqs.~(\ref{eq:Jacobi-Anger1}) and~(\ref{eq:Jacobi-Anger2}) at $K(\tau,\epsilon) := \big\lfloor \frac{1}{2} r\left(\frac{e}{2} |\tau|, \frac{5}{4} \epsilon\right) \big\rfloor$ yields $\epsilon$-approximations to $\cos(x\tau)$ and $\sin(x\tau)$, respectively, where $0 < \epsilon < 1/e$~\cite{Gily_n_2019}. Because $T_i(x)$ is a polynomial of degree $i$ with definite parity, these approximations are polynomials of degree $2K$ and $2K+1$, respectively, with the correct even and odd parity. Let us denote these polynomials by $P^{\text{cos}}_{\epsilon}(x;\tau)$ and $P^{\text{sin}}_{\epsilon}(x;\tau)$.

    Moreover, because cosine and sine are bounded in magnitude by $1$, these $\epsilon$-approximations only obey $|P^{\text{cos}}_{\epsilon}(x;\tau)|, |P^{\text{sin}}_{\epsilon}(x;\tau)| \leq 1+\epsilon$. However, a QSP polynomial must be bounded in magnitude by $1$, which we may force by rescaling these polynomials by a factor of $\frac{1}{1+\epsilon}$, at the expense of increasing the error of these approximations to $2 \epsilon$. This can be seen with the triangle inequality as 
    \begin{equation}
    \begin{aligned}
        &\left|\tfrac{1}{1+\epsilon} P^{\text{cos}}_{\epsilon}(x;\tau) - \cos(x\tau) \right| \\
        &\quad \leq \tfrac{1}{1+\epsilon} \big( \left|P^{\text{cos}}_{\epsilon}(x;\tau) - \cos(x\tau)\right| + \left|\epsilon \cos(x\tau) \right| \big) \\
        &\quad \leq \tfrac{1}{1+\epsilon}(\epsilon + \epsilon) \leq 2\epsilon, 
    \end{aligned}
    \end{equation} 
    and similarly for $P^{\text{sin}}_{\epsilon}(x;\tau)$. 
    
    As we desire an $\epsilon$-approximation to the complex exponential, we should use truncations of the Jacobi-Anger expansion that are $\epsilon/4$-approximate such that, when rescaled by $\frac{1}{1+\epsilon/4}$, they are $\epsilon/2$-approximations to $\cos(x\tau)$ and $\sin(x\tau)$. With this choice, the sum of these approximations, which is the approximation to $e^{-ix\tau}$, is $\epsilon$-approximate by the triangle inequality. Therefore, recalling that our effective goal is to simulate the rescaled Hamiltonian $\mathcal{H}/\alpha$ for an effective time $\tau = \alpha t$, the polynomials of interest are $\frac{1}{1+\epsilon/4} P^{\text{cos}}_{\epsilon/4}(x;\alpha t)$ and $\frac{1}{1+\epsilon/4} P^{\text{sin}}_{\epsilon/4}(x;\alpha t) $. 
    
     Incorporating these conditions, we see that conventional QSP-LCU Hamiltonian simulation queries $\mathcal{H}$ a total number of times 
    \begin{equation}
        \begin{split}
            &2K\left(\alpha t , \frac{\epsilon}{4}\right) + 2K\left(\alpha t, \frac{\epsilon}{4}\right)+1 \\
            &\quad = 4 \cdot \Bigg\lfloor \frac{1}{2} r\left(\frac{e}{2} \alpha |t|, \frac{5}{4} \frac{\epsilon}{4} \right) \Bigg\rfloor +1 =: N_{\mathcal{H}}^{\text{LCU}}(\epsilon; t, \alpha),  \\ 
        \end{split}
        \label{nh0}
    \end{equation}
    where we have defined $N_{\mathcal{H}}^{\text{LCU}}(\epsilon; t, \alpha)$ to denote the sufficient number of queries to $\mathcal{H}$. Evidently, $N_{\mathcal{H}}^{\text{LCU}}(\epsilon; t, \alpha)$ scales asymptotically as 
    \begin{equation}
        N_{\mathcal{H}}^{\text{LCU}}(\epsilon; t, \alpha) = \Theta\left(\alpha |t| + \frac{\ln(1/\epsilon)}{\ln(e+\ln(1/\epsilon)/(\alpha |t|))} \right).
    \end{equation} 
    In comparing this query complexity with results quoted in the literature, $\alpha$ may be replaced with $\| \mathcal{H}\|$.

    \subsection{The Quest for Fully-Coherent Simulation}
    \label{sec:full-coherent-aa}
    As we emphasized at lengths in the introduction, we desire our simulation algorithm to be efficient and fully coherent. Precisely, we define \emph{efficient} simulation algorithms as those whose query complexities are polynomial in $t$, $\ln(1/\epsilon)$, and $\ln(1/\delta)$:
    \begin{customdef}{1}[Efficient Hamiltonian Simulation]\label{def:efficient}
    An Efficient Hamiltonian Simulation algorithm is an algorithm that queries the Hamiltonian a total number of times $\mathcal{O}\Big(\text{poly}\big( \|\mathcal{H} \| |t|, \ \ln(1/\epsilon), \ \ln(1/\delta)\big) \Big)$.
    \end{customdef}
    Furthermore, fully-coherent algorithms perform accurate simulation while requiring only a single copy of the initial state:
    \begin{customdef}{2}[Fully-Coherent Hamiltonian Simulation]\label{def:coherence}
    A Fully-Coherent Hamiltonian Simulation algorithm is an algorithm that, provided a single copy of an initial state $|\psi_0\rangle$, prepares a time evolved state $|\psi\rangle$ such that $\big\| |\psi\rangle - e^{-i\mathcal{H}t} |\psi_0\rangle \big\| \leq \epsilon$ for a generic time-independent Hamiltonian $\mathcal{H}$, with success probability at least $1-\delta$, and for arbitrarily small $\epsilon$ and $\delta$. 
    \end{customdef} 
    This definition of fully-coherent simulation disallows the use of repetition to boost the success probability, which necessarily requires multiple copies of the initial state, and moreover is unsuitable for concatenation into larger algorithms. In addition, this definition applies to generic time-independent Hamiltonians, as we have placed no restrictions on its properties (sparsity, locality, etc.). The algorithms we present will of course assume access to the Hamiltonian $\mathcal{H}$ in some format (i.e., through a block encoding), but it is otherwise left unrestricted. Lastly, we note that in the context of Hamiltonian simulation, we will often use ``fully-coherent" and ``coherent" interchangeably.
    
    Let us now aim to make QSP-LCU Hamiltonian Simulation fully-coherent as per Definition~\ref{def:coherence} by augmentation with amplitude amplification. While such an amplitude amplification procedure is implicitly performed in QSP-based algorithms~\cite{Gily_n_2019}, its formal analysis is often omitted, thus disregarding the coherence of the algorithm. We first discuss how full-coherence may be achieved with conventional QSP-based amplitude amplification at the expense of a multiplicative factor of $\Theta(\ln(1/\delta))$ appended to the query complexity, and subsequently how robust oblivious amplitude amplification may alternatively be employed to contribute only an additive term of $\Theta(\ln(1/\delta))$.

    \subsection{Achieving Full Coherence with Conventional Amplitude Amplification}\label{sec:conv_ampamp}
    In this section, we discuss how to transform the incoherent simulation algorithm of the previous section into a fully-coherent algorithm by integration with amplitude amplification. We first describe the necessary amplitude amplification procedure in Sec.~\ref{sec:ampamp-procedure}, and then present a query complexity analysis in Sec.~\ref{sec:ampamp-complexity}. We will refer to the algorithm presented in this section as the ``QSP-LCU + Amplitude Amplification" method, or ``QSP-LCU+AA" for short.
    
    \subsubsubsection{Procedure}\label{sec:ampamp-procedure}
    Recall that the output of conventional QSP-LCU Hamiltonian simulation is an $\epsilon/2$-approximation to $\frac{1}{2}e^{-i\mathcal{H}t}$, which we will denote by $A$. $A$ has amplitude $\frac{1-\epsilon}{2} \leq \|A\| \leq \frac{1+\epsilon}{2}$, which we would like to increase to a value at least $\sqrt{1-\delta}$, such that the probability of failure (i.e. accessing the wrong block) is $\leq \delta$. 
    
    We will aim to achieve this amplification with a QSP-based procedure, noting that the QSP-LCU circuit of Fig.~\ref{fig:ComplexExponentialCircuit} is a unitary block encoding of $A$ with projector $\Pi = |00\rangle \langle 00|$. Our desired amplification is nontrivial to achieve using QET because the eigenvalues of $A$ are necessarily complex, making it difficult to find a polynomial that will amplify just their magnitudes; for instance, the uniform spectral amplification polynomial of Ref. \cite{Low_2017_spectral} applies only to real-valued inputs and will not suffice here. Such amplitude amplification necessitates the use of QSVT, a statement that we make rigorous with the following lemma:
    \begin{lemma}[Eigenvalue Magnitude Transformation]\label{lemma:magnitude_tr}
    Given a unitary block-encoding of $A= \sum_k \lambda_k |\lambda_k\rangle \langle \lambda_k|$, where the eigenvalues $\lambda_k = r_k e^{i\theta_k}$ may be complex, one can polynomially transform the magnitudes $r_k$ while retaining the phases $e^{i\theta_k}$, by applying QSVT to the encoding with an \textit{odd degree} polynomial $P(x)$. The output of this protocol is an encoding of a matrix $\tilde{A}= \sum_k P(r_k) e^{i\theta_k} |\lambda_k\rangle \langle \lambda_k|$, whose eigenvalues have the same phases as those of $A$, but magnitudes transformed by $P(x)$. 
    \end{lemma} 
    \begin{proof}
    First observe that $A$ has singular value decomposition $A = \sum_k \sigma_k |w_k\rangle \langle v_k|$, where $\sigma_k = r_k$ and $|w_k\rangle \langle v_k| = e^{i\theta_k} |\lambda_k\rangle \langle \lambda_k| $. This expression makes it evident that the magnitudes of the eigenvalues are encoded in the singular values $\sigma_k$, while the phases are stored in the singular vector product $|w_k\rangle \langle v_k|$.
    
    Therefore, to transform the magnitudes, one may apply to the encoding of $A$ a singular value transformation with the desired polynomial. However, to ensure that the phases $e^{i\theta_k}$ are preserved by this transformation, we must employ QSVT with an odd degree polynomial, such that the output will be $\sum_k P(\sigma_k) |w_k\rangle \langle v_k| = \sum_k P(r_k)e^{i\theta_k} |\lambda_k\rangle \langle \lambda_k| = \tilde{A}$ (see  Appendix~\ref{sec:QSVT_overview}), which transforms the magnitudes but retains the phases. On the other hand, an even degree polynomial would output $\sum_k P(\sigma_k) |v_k\rangle \langle v_k| = \sum_k P(r_k) |\lambda_k\rangle \langle \lambda_k|$, which does not retain the phases.
    \end{proof}
    
    We may apply this lemma to our advantage. First note that if $\mathcal{H}$ has eigenvalue decomposition $\mathcal{H} = \sum_k E_k  |E_k\rangle \langle E_k|$, then $A$ has singular value decomposition $A = \sum_k \sigma_k |w_k\rangle \langle v_k|$ where $\frac{1-\epsilon}{2} \leq \sigma_k \leq \frac{1+\epsilon}{2}$, and $\big\| |w_k\rangle \langle v_k| - e^{-i E_k t}  |E_k\rangle \langle E_k| \big\| \leq \epsilon$ (this last inequality comes from the worst case scenario when $\sigma_k=1/2$ and all of the errors are due to incorrect phases in $|w_k\rangle \langle v_k|$). As we would like to amplify $\|A\|$ to a value at least $\sqrt{1-\delta}$, it will suffice to choose a polynomial $P(x)$ that maps inputs $\geq \frac{1-\epsilon}{2}$ to a value at least $1-\delta/2 \geq \sqrt{1-\delta}$, which will guarantee a probability of failure $\leq \delta$. Drawing inspiration from Refs.~\cite{Gily_n_2019, martyn2021grand}, we may select $P(x)$ to be a polynomial approximation to a sign function. 
    
    We denote the sign function by $\text{sign}(x)$, and it obeys
    \begin{equation}
        \text{sign}(x) = \begin{cases}
            -1 & x<0 \\
            0 & x=0 \\
            1 & x>0. 
        \end{cases}
    \end{equation}
    In Appendix~\ref{Sec:SignFunction}, we illustrate a construction of a polynomial approximation to $\text{sign}(x)$ that is accurate away from the discontinuity at $x=0$. This polynomial, which we denote by $P^\text{sign}_{\epsilon, \Delta}(x)$, obeys 
    \begin{equation}
        \big|P^\text{sign}_{\epsilon, \Delta}(x) - \text{sign}(x) \big| \leq \epsilon \ \text{ for } x\in [-1,1] \backslash \left[-\Delta/2, \ \Delta/2 \right].
    \end{equation}
    That is, for $x$ a distance at least $\Delta/2$ away from the discontinuity at $x=0$, $P^\text{sign}_{\epsilon, \Delta}(x)$ provides an $\epsilon$-approximation to the sign function. We prove in Appendix~\ref{Sec:SignFunction} that $P^\text{sign}_{\epsilon, \Delta}(x)$ can be constructed as a polynomial of odd degree $d=\gamma(\epsilon, \Delta)$, where $\gamma(\epsilon, \Delta)$ is a function explicitly defined in Eq.~(\ref{eq:gamma}) that scales as $\Theta\left(\frac{1}{\Delta} \ln(\frac{1}{\epsilon})\right)$. As this degree is necessarily odd, Lemma~\ref{lemma:magnitude_tr} applies.

    \begin{figure*}[htbp]
        \begin{center}
        \includegraphics[width=18cm]{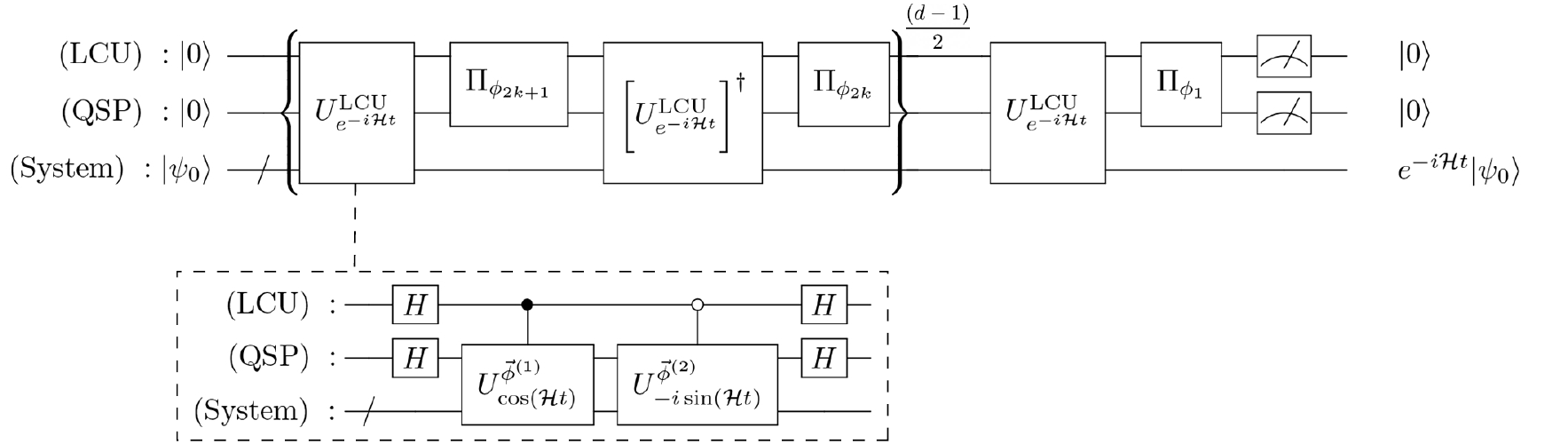}
        \end{center}
        \caption{A quantum circuit depicting the QSP-LCU+AA fully-coherent simulation algorithm. Here, our projectors are $\Pi = |00\rangle\langle 00| = \tilde{\Pi}$, and the corresponding projector-controlled phase shift is $\Pi_{\phi_k}$ where the phase $\phi_k$ is applied to the subspace $\ket{00} \bra{00}$. The gates within brackets are repeated $(d-1)/2$ times for $k = (d-1)/2, (d-1)/2-1, ..., 1$ with phases $\vec{\phi}$ to construct the appropriate QSVT sequence for amplitude amplification. This procedure succeeds if the two ancilla qubits (previously dubbed the LCU and QSP qubits) are measured in state $|00\rangle$, which occurs with high probability $\geq 1-\delta$ after appropriate amplitude amplification. We also note that the unitary $U^{\rm{LCU}}_{e^{-i\mathcal{H}t}}$ is defined by the circuit in the inset, which is identical to the circuit from Fig.~\ref{fig:ComplexExponentialCircuit} used in QSP-LCU simulation.}
    \label{fig:AmplitudeAmpSimulationCircuit}
    \end{figure*}

    \begin{figure*}[htbp]
    \begin{center}
    \includegraphics[width=15cm]{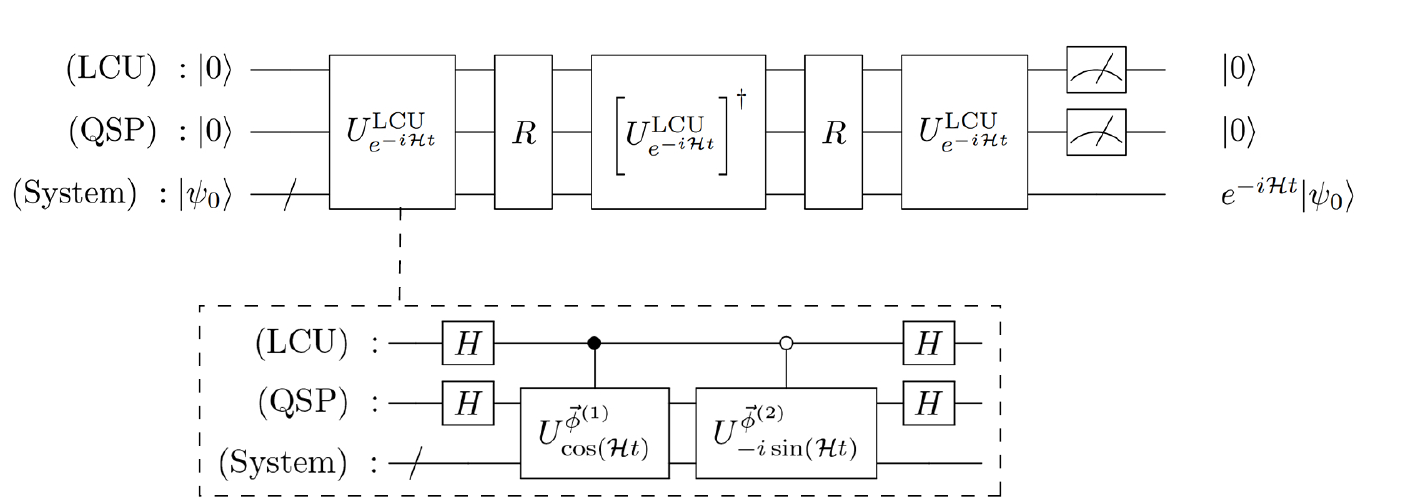}
    \end{center}
    \caption{A quantum circuit depicting the QSP-LCU+ROAA protcol, where $R = I - 2P$ (for projector $P = |00\rangle \langle 00| \otimes I$) is an ancilla operator as in Eq.~\eqref{oaa-pw}. Like the circuit  of Fig.~\ref{fig:AmplitudeAmpSimulationCircuit}, this procedure succeeds if the LCU and QSP qubits  are measured in state $|00\rangle$, which occurs with high probability $\geq 1-\delta$ after robust oblivious amplitude amplification.}
    \label{fig:ObliviousAmplitudeAmpSimulationCircuit}
    \end{figure*}

    Therefore, to perform our sought-after amplitude amplification, it will suffice to choose the QSVT polynomial to be
    \begin{equation}
        P(x) = P^\text{sign}_{\delta/2, (1-\epsilon)}(x),
    \end{equation}
    which will indeed map inputs $x>\frac{1-\epsilon}{2}$ to values $\geq 1-\delta/2 \geq \sqrt{1-\delta}$, as desired. Applying this to $A$, we see that the output will be an $\epsilon$-approximation to $e^{-i\mathcal{H}t}$ that is accessed with probability at least $1-\delta$. We depict in Fig.~\ref{fig:AmplitudeAmpSimulationCircuit} the quantum circuit that implements this fully-coherent simulation via amplitude amplification, highlighting how it applies QSVT to the unitary resulting from conventional QSP-LCU Hamiltonian simulation.

    \subsubsubsection{Query Complexity} \label{sec:ampamp-complexity}
    As this procedure applies QSVT to $A$, which is already constructed as an eigenvalue transformation, we find that the number of queries to $\mathcal{H}$ is a product of the degrees of the two polynomials used in these transformations:
    \begin{equation}\label{eq:NAmp}
        N^{\text{AA}}_{\mathcal{H}}(\epsilon, \delta; t, \alpha) = \gamma(\delta/2, 1-\epsilon) N^{\text{LCU}}_{\mathcal{H}}(\epsilon; t, \alpha).
    \end{equation}
    (the superscript ``AA" standing for amplitude amplification). Using the asymptotic behavior of $\gamma(\delta/2, 1-\epsilon)$ and  $N^{\text{LCU}}_{\mathcal{H}}(\epsilon; t, \alpha)$, we find that $N^{\text{AA}}_{\mathcal{H}}(\epsilon, \delta; t, \alpha)$ scales as 
    \begin{equation}
    \begin{aligned}
        &N^{\text{AA}}_{\mathcal{H}}(\epsilon, \delta; t, \alpha) = \\
        & \qquad \quad \Theta\left( \ln(\frac{1}{\delta}) \left(\alpha |t| + \frac{\ln(1/\epsilon)}{\ln(e+\ln(1/\epsilon)/(\alpha |t|))} \right) \right).
    \end{aligned}
    \end{equation}
    Although efficient by Def.~\ref{def:efficient}, this query complexity has a $\Theta \left( \ln(\frac{1}{\delta}) \right)$ multiplicative prefactor, which can be costly. This result evidences the following theorem: 
    \begin{theorem}[QSP-LCU+AA Hamiltonian Simulation]\label{thm:qsp-lcu-aa}
    Provided a block encoding of the Hamiltonian (possibly rescaled), the QSP-LCU+AA simulation algorithm achieves efficient fully-coherent Hamiltonian simulation with arbitrarily small $\epsilon$ and $\delta$, while querying (a block encoding of) the Hamiltonian a total number of times 
    \begin{equation}
        \Theta\left( \ln(\frac{1}{\delta}) \left(\alpha |t| + \frac{\ln(1/\epsilon)}{\ln(e+\ln(1/\epsilon)/(\alpha |t|))} \right) \right).
    \end{equation}
    \end{theorem}

    \subsection{Achieving Full Coherence with Robust Oblivious Amplitude Amplification}\label{sec:RO_ampamp}
    An alternative approach to achieving full coherence is through robust oblivious amplitude amplification, which we achieve in the current section. We will refer to the resulting algorithm as the ``QSP-LCU + Robust Oblivious Amplitude Amplification" method, or ``QSP-LCU+ROAA" for short.
    
    Earlier, we established that the conventional QSP-LCU simulation algorithm only succeeds with probability $\frac{(1-\epsilon)^2}{4} \leq p \leq \frac{(1+\epsilon)^2}{4} $. If the success probability were known exactly (say $p=\frac{1}{4}$ precisely), one could use the amplification procedure of Ref.~\cite{Brasard_2002} to boost the success probability arbitrarily high while appending only an $\mathcal{O}(1)$ factor to the query complexity. However, as we only have upper and lower bounds on the probability of success here, one can alternatively employ the robust oblivious amplitude amplification protocol from Ref.~\cite{2015Berry} to boost the success probability, which also appends an $\mathcal{O}(1)$ factor to the query complexity. The resulting error scales as $\delta = \Theta(\epsilon)$, which renders the query complexity additive in $\ln(1/\delta)$, and thus on par with our coherent one-shot simulation algorithm of Sec.~\ref{Sec:OneShotHamSim}.

    \subsubsubsection{Procedure}
    Robust oblivious amplitude amplification pertains to the following scenario: given a unitary $W$ whose application to a state $|\psi\rangle$ produces a block encoding of $\tilde{U} |\psi \rangle$ with magnitude close to $1/2$ for some matrix $\tilde{U}$ close to a unitary $U_r$, we seek to construct an operator $A$ whose application to $| \psi \rangle$ produces such a block encoding with magnitude close to $1$. More precisely, given $W$ and a projector $P = |0\rangle \langle 0| \otimes I$ such that 
    \begin{equation}
        PW|0\rangle |\psi \rangle = \frac{1}{s} |0\rangle \tilde{U} |\psi \rangle 
        \label{oaa-pw}
    \end{equation}
    where $|s-2|= \mathcal{O}(\delta)$ and $\|\tilde{U} - U_r\| = \mathcal{O}(\delta)$, the robust oblivious amplitude amplification protocol constructs an operator $A$ that satisfies 
    \begin{equation}
        \big\| PA |0\rangle |\psi \rangle - |0\rangle U_r |\psi \rangle \big\| = \mathcal{O}(\delta).
    \end{equation}
    In other words, $A$ applies an approximation to the unitary $U_r$ with error order $\delta$. $A$ is explicitly constructed as $A=-WRW^\dag RW$, where $R=I-2P$ is an ancilla reflection operator.
    
    This construction is directly applicable to the amplitude amplification necessary for QSP-based simulation, the output of which is a QSP sequence that block encodes an $\epsilon$-approximation to $e^{-i\mathcal{H}t}$ with amplitude $\epsilon$-close to $\frac{1}{2}$. Therefore, by invoking this procedure with QSP-LCU circuit as $W$ and $P=|00\rangle \langle 00| \otimes I$, we may produce a block encoding of an $\mathcal{O}(\epsilon)$-approximation to $e^{-i\mathcal{H}t}$. We illustrate in Fig.~\ref{fig:ObliviousAmplitudeAmpSimulationCircuit} the quantum circuit that implements fully-coherent simulation via robust oblivious amplitude amplification.
    
    \subsubsubsection{Query Complexity}
    Mapping the QSP-LCU Hamiltonian simulation to the setting of robust oblivious amplitude amplification, we may take $U_r = e^{-i \mathcal{H} t}$, $s=2$, and $\| \tilde{U} - U_r\| \leq \epsilon $, which corresponds to a block encoding of some operator $\frac{1}{2}(e^{-i \mathcal{H}t} \pm \epsilon)$, as desired. This dictates that 
    \begin{equation}
    \begin{aligned}
        &\| \tilde{U} \tilde{U}^\dag - I \| \leq \|(U_r \pm \epsilon)(U_r^\dag \pm \epsilon) - I \| \leq \\
        & \delta \| U_r + U_r^\dag \| + \epsilon^2 \leq 3\epsilon.
    \end{aligned}
    \end{equation}
    Ref.~\cite{Berry_2015} shows that $PA |0\rangle |\psi\rangle = | 0 \rangle (\frac{3}{s}\tilde{U} - \frac{4}{s^3} \tilde{U} \tilde{U}^\dag \tilde{U}) |\psi \rangle$, which in our case lends itself to the inequality  
    \begin{equation}
    \begin{split}
        &\left\| \tilde{U} \left(\frac{3}{s} - \frac{4}{s^3} \tilde{U}^\dag \tilde{U} \right) - U_r \right\| \leq \\
        &\qquad \left\| \tilde{U} \left(\frac{3}{2} - \frac{1}{2} (I\pm 3\epsilon)\right) - U_r \right \| \leq \\ 
        &\qquad \frac{3}{2}\epsilon + \| \tilde{U} - U_r \| \leq \frac{5}{2}\epsilon.
    \end{split}
    \end{equation}
    Therefore, oblivious robust amplitude amplification applied to QSP-LCU Hamiltonian simulation achieves an error bounded by $\frac{5}{2}\epsilon$. This implies that the correct block of the block encoding has magnitude $\geq 1-\frac{5}{2}\epsilon$, and that the corresponding failure probability of this algorithm is bounded by $\delta \leq 1-(1-\frac{5}{2}\epsilon)^2 \leq 5 \epsilon$.
    
    If we are interested then in the query complexity of achieving fully-coherent simulation with a total error $\epsilon$, we may take $\epsilon \mapsto \frac{2}{5}\epsilon$ in the equations of Sec.~\ref{Sec:Conventional_Complexity}. Observing that the robust oblivious amplitude amplification operator $A$ employs three instances of the QSP-LCU circuit $W$, the total query complexity is 
    \begin{equation}\label{eq:N_OAA}
        N_{\mathcal{H}}^{\text{ROAA}}(\epsilon; t, \alpha) = 3 N_{\mathcal{H}}^{\text{LCU}}(\tfrac{2}{5}\epsilon; t, \alpha)
    \end{equation}
    (``ROAA" for robust oblivious amplitude amplification). As above, the corresponding probability of failure is bounded as $\delta \leq 2\epsilon$. We thus have:
    \begin{theorem}[QSP-LCU+ROAA Hamiltonian Simulation]\label{thm:qsp-lcu-aa}
    Provided a block encoding of the Hamiltonian (possibly rescaled), the QSP-LCU+ROAA simulation algorithm achieves efficient fully-coherent Hamiltonian simulation with arbitrarily small $\epsilon$ and $\delta=2\epsilon$, while querying (a block encoding of) the Hamiltonian a total number of times 
    \begin{equation}
        \Theta\left( \alpha |t| + \frac{\ln(1/\epsilon)}{\ln(e+\ln(1/\epsilon)/(\alpha |t|))} \right).
    \end{equation}
    \end{theorem}

\section{Coherent One-Shot Hamiltonian Simulation} \label{Sec:OneShotHamSim}

In Sec.~\ref{Sec:ConventionalHamSim}, we presented a simple but non-optimal fully-coherent Hamiltonian simulation algorithm that acquired a multiplicative prefactor of $\ln \left( 1 / \delta \right)$ in its query complexity, and subsequently showed how to remove this factor by using robust oblivious amplitude amplification \cite{2015Berry} to achieve fully-coherent simulation. 

In this section, we introduce our novel ``coherent one-shot Hamiltonian simulation" algorithm that attains a query complexity \emph{additive} in $\ln \left( 1 / \delta \right)$ rather than multiplicative, and is thus also both fully-coherent and efficient. In contrast to the algorithms in Sec.~\ref{Sec:ConventionalHamSim}, our new construction guarantees the unitarity of the simulation algorithm by using a single QSP call to directly approximate the complex exponential function $e^{-i x \tau}$, without the need for any amplitude amplification technique. We will discuss the performance trade-off between our new coherent one-shot simulation algorithm and the two amplitude amplification based fully-coherent algorithms in detail in Sec.~\ref{Sec:Scaling}.

We outline the coherent one-shot algorithm in Sec. \ref{sec:coherent-procedure}, where details for the block-encoding (Sec.~\ref{sec:coherent-be}), the construction of the complex exponential function (Sec. \ref{sec:coherent-eece}), as well as the success probability (Sec. \ref{Sec:Coherent_Success}) are presented. This is followed by a query complexity analysis in Sec. \ref{Sec:Coherent_Runtime} and a summary of the main result in Theorem~\ref{thm:coherent-one-shot}.

\subsection{Procedure} \label{sec:coherent-procedure}
Coherent one-shot Hamiltonian simulation is achieved by taking a modified approach to QSP: we first apply a transformation to the Hamiltonian, realized via unitary gates and ancilla qubits, and subsequently apply QSP to the transformed Hamiltonian. We call this first step a \textit{pre-transformation}: it rescales the eigenvalues of the Hamiltonian, which enables us to circumvent the parity constraint on QSP polynomials and ultimately achieve a broader class of transformations. We also note that as the Hamiltonian is accessed via a block encoding, the output of the pre-transformation is really a block-encoding of the transformed Hamiltonian, possibly with a different projector. For the sake of simplicity, we will employ a linear pre-transformation in the development of coherent one-shot Hamiltonian simulation; however, the pre-transformation may be nonlinear in general.

In more detail, our starting point is a block encoding of $\mathcal{H}/\alpha$ with eigenvalues in the range $[-1,1]$. In the coherent one-shot algorithm, we first find a linear pre-transformation that rescales its eigenvalues into the range $(a,b) \subset [-1,1]$, and then determine a QSP polynomial that is a good approximation to the complex exponential $e^{-ix\tau}$ over the range $x \in (a,b) \subset[-1,1]$. Finally, by applying a quantum eigenvalue transformation (QET) to this encoding with our polynomial approximation to the complex exponential and an appropriately chosen effective time $\tau$, we attain the time evolution operator up to a global phase. Note how the use of the pre-transformation shifts our focus to eigenvalues only in the reduced range $(a,b) \subset [-1,1] $, which allows us to ignore the boundary conditions at $x= \pm 1$ imposed by QSP. 

Below, we specialize to one particular implementation of this method, employing a linear pre-transformation and a polynomial that approximates the complex exponential $e^{-ix\tau}$ over a range of positive $x$: $x \in (a,b) \subset[0,1]$. While this construction is sufficient for demonstrating the advantage of coherent one-shot simulation, it is by no means optimal as alternate constructions, including other pre-transformations, certainly exist. 

In general, the use of a pre-transformation provides an alternative to the single-ancilla QSP algorithm of Ref.~\cite{Low_2019} for the construction of the complex exponential, and even of more general functions of mixed parity. Whereas single-ancilla QSP imposes constraints on the QSP phases, which both complicates the optimization over QSP phases and inhibits its robustness to noise, our method using a pre-transformation is less restricted. For example, employing PyQSP~\cite{pyqsp} to numerically estimate the complex exponential with $\tau=4$ with a length $d=14$ QSP sequence, single-ancilla QSP can be optimized to achieve an error $0.048$ (i.e. the maximal difference from the desired function over the range of interest), while our method (with $a=1/4, \ b=3/4$) obtains an error $0.027$. And furthermore, when the phases are subject to additive Gaussian noise of standard deviation $0.02$, the error of single-ancilla QSP grows to $0.085 \pm 0.029$, while that of our method becomes only $0.067 \pm 0.025$.

\subsubsection{Pre-transformation} \label{sec:coherent-be}
With our encoding of $\mathcal{H}/\alpha$, we will employ a linear pre-transformation to block encode an operator whose spectrum is proportional to that of $\mathcal{H}/\alpha$ (up to an additive constant), but shrunken to be in the range $[\frac{1-\beta}{2}, \frac{1+\beta}{2}] \subset [0,1]$ for some chosen $\beta < 1$.

To achieve this, we first block encode a rescaled Hamiltonian $\beta \mathcal{H}/\alpha$. Given the ability to block encode $\mathcal{H}/\alpha$, which is already scaled by a constant itself, it is not much more difficult to block encode $\beta \mathcal{H}/\alpha$, and hence this step may not even be necessary. If this step is needed however, it may formally be achieved by invoking Lemma 53 of Ref. \cite{Gily_n_2019} regarding products of block encoded matrices: 
\begin{lemma}\label{lemma:encoding}
If $V_A$ is a block encoding of an $n$ qubit operator $A$, and if $V_B$ is a block encoding of an $n$ qubit operator $B$, then $(I_B\otimes V_A)(V_B\otimes I_A)$ is a block encoding of $AB$, where the identity operators in this expression act on each other's ancilla spaces.
\end{lemma}

This result, which is straightforwardly proven via direct computation, enables our construction. Suppose that $U_{\mathcal{H}/\alpha}$ is the block encoding of the $n$ qubit operator $\mathcal{H}/\alpha$, and that $U_{\beta I}$ is a block encoding of $\beta I_{2^n}$. Then, using the formula of this lemma, we may construct a block encoding of $\mathcal{H}/\alpha \cdot \beta I_{2^n} = \beta \mathcal{H}/\alpha$, as desired. Moreover, we may easily construct $U_{\beta I}$ with an $x$-rotation applied to an ancilla qubit:
\begin{equation}
\begin{split}
        R_x(2\cos^{-1}(\beta)) \otimes I_{2^n} = 
        \begin{bmatrix}
            \beta I_{2^n} & \cdot \\
            \cdot & \cdot 
        \end{bmatrix} = U_{\beta I},
\end{split}
\end{equation}
where $R_x(\theta)$ is the $x$-rotation through an angle $\theta$. Equivalently, the block encoding of $\beta \mathcal{H}/\alpha$ is $U_{\beta \mathcal{H}/\alpha} = R_x(2\cos^{-1}(\beta)) \otimes U_{\mathcal{H}/\alpha}$, which we depict in Fig.~\ref{fig:HamScalingEncoding}. As QSP already employs a sequence of rotations to develop $U^{\vec{\phi}}$, this additional x-rotation is not particularly costly.

\begin{figure}[htbp]
    \begin{center}
    \includegraphics[width=6.2cm]{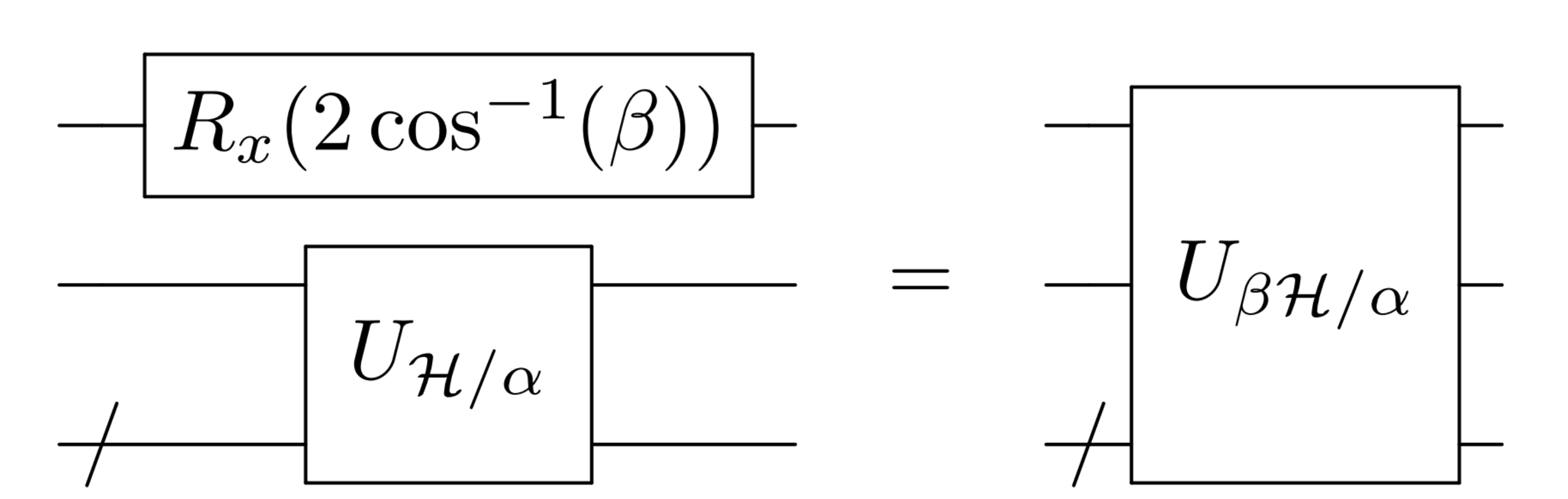}
    \end{center}
    \caption{A quantum circuit representation of the operator $U_{\beta\mathcal{H}/\alpha}$, which block encodes $\beta \mathcal{H}/\alpha$ in its $|00\rangle \langle 00 | $ matrix element. Here we have assumed that $U_{\mathcal{H}/\alpha}$ encodes $\mathcal{H}/\alpha$ in its $|0\rangle \langle 0|$ matrix element, as is conventional.}
    \label{fig:HamScalingEncoding}
\end{figure}

With this rescaled Hamiltonian, we may introduce an ancilla qubit and employ the circuit in Fig.~\ref{fig:HamSumEncoding} to obtain an encoding of $\frac{1}{2}(I+\beta \mathcal{H}/\alpha) =: \tilde{\mathcal{H}}$, which has eigenvalues in the range $[\frac{1-\beta}{2}, \frac{1+\beta}{2}] \subset [0,1]$. This is precisely our desired linear pre-transformation: $\mathcal{H}/\alpha \mapsto \frac{1}{2}(I+\beta \mathcal{H}/\alpha)$. In total then, with the addition of two extra ancilla qubits, we are able to block encode the rescaled Hamiltonian $ \tilde{\mathcal{H}} = \frac{1}{2}(I+\beta \mathcal{H}/\alpha)$ which has sufficiently bounded eigenvalues. As for accessing this operator, if the initial Hamiltonian $\mathcal{H}/\alpha$ were encoded in the $|0\rangle \langle 0|$ matrix element of a unitary, as is conventional, then the procedure sketched here encodes our desired operator in the $|000\rangle \langle 000|$ matrix element of a new unitary.

\begin{figure}[htbp]
    \begin{center}
    \includegraphics[width=5.8cm]{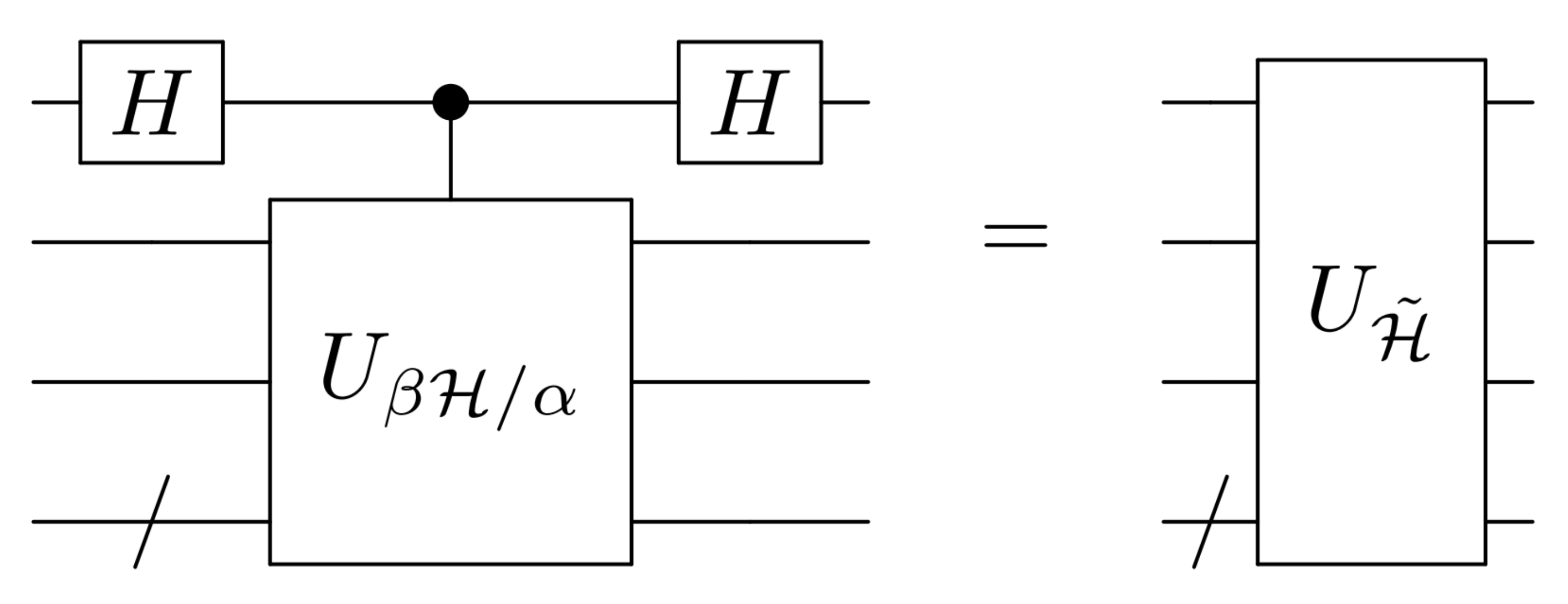}
    \end{center}
    \vspace{-10px}
    \caption{A quantum circuit representation of the operator $U_{\tilde{\mathcal{H}}}$, which block encodes $\tilde{\mathcal{H}} = \frac{1}{2}(I+\beta \mathcal{H}/\alpha)$ in its $|000\rangle \langle 000 |$ matrix element.}
    \label{fig:HamSumEncoding}
\end{figure}

\subsubsection{Target Polynomial} \label{sec:coherent-eece}
With an eye towards Hamiltonian simulation, we aim to approximate $e^{-ix\tau}$ as a polynomial, where $x$ is the input variable and $\tau$ is a real parameter representing the effective time of simulation. Unfortunately, as we discussed in Sec.~\ref{Sec:ConventionalHamSim}, we cannot approximate $e^{-ix\tau}$ as a QSP polynomial $\forall x\in[-1,1]$ because the complex exponential does not have definite parity. However, using the block encoding of $\tilde{\mathcal{H}}$ outlined above, we only need a polynomial that approximates $e^{-ix\tau}$ for $x\in [\frac{1-\beta}{2}, \frac{1+\beta}{2}] \subset [0,1]$. We may design such a polynomial by estimating the \textit{even extension of the complex exponential} (EECE):
\begin{equation}\label{eq:EECE}
    \text{EECE}(x;\tau) := \cos(\tau x)-i\sin(\tau x)\text{sign}(x).
\end{equation}
This function has definite parity, and is of magnitude $1$, so there exists a QSP polynomial $P(x) \approx \text{EECE}(x;\tau)$ which can be accessed as the $|0\rangle \langle 0| $ matrix element of a QSP sequence. In approximating this function as a polynomial, we may estimate the trigonometric functions with the truncated Jacobi-Anger expansions discussed in Sec.~\ref{Sec:ConventionalHamSim}, and employ the polynomial approximation to the sign function explained in Appendix~\ref{Sec:SignFunction}.

Formally, it suffices to select the following as a polynomial approximation to the EECE:
\begin{equation} \label{eq:OneShotPoly}
    \begin{split}
       &P_{\epsilon,\Delta}^{\text{EECE}}(x;\tau) := \\
       &\qquad \tfrac{1}{1+\epsilon/6}P^{\text{cos}}_{\epsilon/6}(x;\tau)- i \tfrac{1}{1+\epsilon/6}P^{\text{sin}}_{\epsilon/6}(x;\tau) P^{\text{sign}}_{\epsilon/3,\Delta}(x).
    \end{split}
\end{equation}
As $\frac{1}{1+\epsilon/6}P^{\text{cos}}_{\epsilon/6}(x;\tau)$ and $\frac{1}{1+\epsilon/6}P^{\text{sin}}_{\epsilon/6}(x;\tau)$ are $\epsilon/3$ approximations to $\cos(\tau x)$ and $\sin(\tau x)$, respectively, the error of this approximation for $|x|\geq \Delta/2$ is 
\begin{equation}
    \begin{split}
        &\left\|\cos(\tau x) - \frac{1}{1+\epsilon/6}P^{\text{cos}}_{\epsilon/6}(x;\tau) \right\| \ + 
        \\ &\left\|\sin(\tau x) - \frac{1}{1+\epsilon/6}P^{\text{sin}}_{\epsilon/6}(x;\tau) P^{\text{sign}}_{\epsilon/3,\Delta}(x) \right\| \\
        &\leq \epsilon/3 + \left\|\sin(\tau x) - (\sin(\tau x)-\epsilon/3)(1-\epsilon/3) \right\| \\
        &\leq \epsilon/3+2\epsilon/3 - (\epsilon/3)^2 \leq \epsilon,
    \end{split}
\end{equation}
as desired. So $P_{\epsilon,\Delta}^{\text{EECE}}(x;\tau)$ $\epsilon$-approximates the even extension of the complex exponential for $|x|\geq \Delta/2$. In addition, $P_{\epsilon,\Delta}^{\text{EECE}}(x;\tau)$ also has magnitude less than $1$ by virtue of its constituent polynomials being bounded in magnitude, as is required by the third condition of the QSP theorem. We rigorously analyze the degree of $P_{\epsilon,\Delta}^{\text{EECE}}(x;\tau)$ in Sec.~\ref{Sec:Coherent_Runtime}.

One concern with a QSP implementation of $P_{\epsilon,\Delta}^{\text{EECE}}(x;\tau)$ is that according to the third QSP condition, $|P(x\rightarrow \pm 1)| \rightarrow 1$, which $P_{\epsilon,\Delta}^{\text{EECE}}(x;\tau)$ only satisfies up to some error $\epsilon$. Fortunately, this is not a severe problem for us because $x=\pm 1$ is outside our range of interest, that being  $\big[\frac{1-\beta}{2} , \frac{1+\beta}{2} \big]$. Likewise, for reasonably small $\epsilon$, there should exist a QSP polynomial that behaves as $P_{\epsilon,\Delta}^{\text{EECE}}(x;\tau)$ and also satisfies the third condition of the QSP theorem. In Sec.~\ref{Sec:Discussion}, we numerically verify that such a QSP polynomial can be constructed.

How might we employ this polynomial? Consider applying a QET to the aforementioned encoding of $\tilde{\mathcal{H}}$ with the target polynomial $P_{\epsilon,\Delta}^{\text{EECE}}(x;\tau)$. If we select a $\Delta/2 \leq (1-\beta)/2$ such that $P_{\epsilon,\Delta}^{\text{EECE}}(x;\tau)$ acts as the complex exponential on the eigenvalues of $\tilde{\mathcal{H}}$, and an effective time $\tau = 2t\alpha/\beta$, the resulting operator of this transform will encode an $\epsilon$-approximation to
\begin{equation}
    e^{-i\frac{1}{2}(I+\beta H/\alpha) 2t\alpha/\beta } = e^{-it\alpha / \beta} e^{-i\mathcal{H}t},
\end{equation}
which is the time evolution operator, up to a global phase. So by using the polynomial $P_{\epsilon,(1-\beta)}^{\text{EECE}}(x;\ 2t\alpha/\beta)$ this procedure yields a unitary $U^{\vec{\phi}}$ and a projector $\Pi$ (which in the simple case mentioned earlier, is just $\Pi = |000\rangle \langle 000|$), such that $\| \Pi U^{\vec{\phi}} \Pi - e^{-i\mathcal{H}t}\| \leq \epsilon$, thus achieving Hamiltonian simulation!

\subsubsection{Probability of Success}\label{Sec:Coherent_Success}
As we mentioned above, we may employ our block encoding of $\tilde{\mathcal{H}}$ and polynomial approximation to the EECE to create a unitary that block encodes an $\epsilon$-approximation to $e^{-i\mathcal{H}t}$. Because $\| e^{-i\mathcal{H}t}\| = 1$, the magnitude of this block is close to $1$, which by unitarity implies that the other elements in its row/column have magnitude near zero, and that the $e^{-i\mathcal{H}t}$ block is accessed with high probability. Specializing to the case where $\mathcal{H}/\alpha$ is encoded in the $|0\rangle \langle 0|$ matrix element of a unitary and the projector of interest is $\Pi = |000\rangle \langle 000|$, this means that the application of $U^{\vec{\phi}}$ to $|000\rangle |\psi_0\rangle$ outputs a good approximation $|000\rangle e^{-i\mathcal{H}t}|\psi_0\rangle$ with high probability. Therefore, this algorithm indeed performs Hamiltonian simulation with high success probability using only a single QET call, hence the aptly chosen name ``coherent one-shot Hamiltonian simulation".

More precisely, as $\| \Pi U^{\vec{\phi}} \Pi - e^{-i\mathcal{H}t} \| \leq \epsilon$, the reverse triangle inequality indicates that the magnitude of this block is
\begin{equation}
    \| e^{-ix\tau} + \epsilon \| \geq 1-\epsilon.
\end{equation}
So unitarity dictates that the other elements in the column corresponding to the choice of projector have collective magnitude $\leq \sqrt{1-(1-\epsilon)^2}$. This implies that the probability of failure is in the worst case
\begin{equation}
    \delta = 1-(1-\epsilon)^2 \leq 2\epsilon \quad \Rightarrow \quad \delta = \Theta(\epsilon).
\end{equation}
As evidenced by this expression, the probability of failure of this algorithm is controlled by the accuracy $\epsilon$ of the EECE approximation, which may easily be tuned by selecting an appropriate QSP polynomial. We note that a similar bound on the probability of failure was derived for the Hamiltonian simulation algorithm presented in Ref. \cite{dong2021quantum}, although their result applies to the simulation of a random Hamiltonian. 

We graphically summarize the coherent one-shot algorithm in Fig.~\ref{fig:OneShotCircuit}, in which we illustrate its quantum circuit. In this figure, we emphasize that this algorithm applies a QET to $\tilde{\mathcal{H}}$ and succeeds when the ancilla qubits are measured in the state $|000\rangle$.

\begin{figure*}[htbp]
    \begin{center}
    \includegraphics[width = 11.5cm]{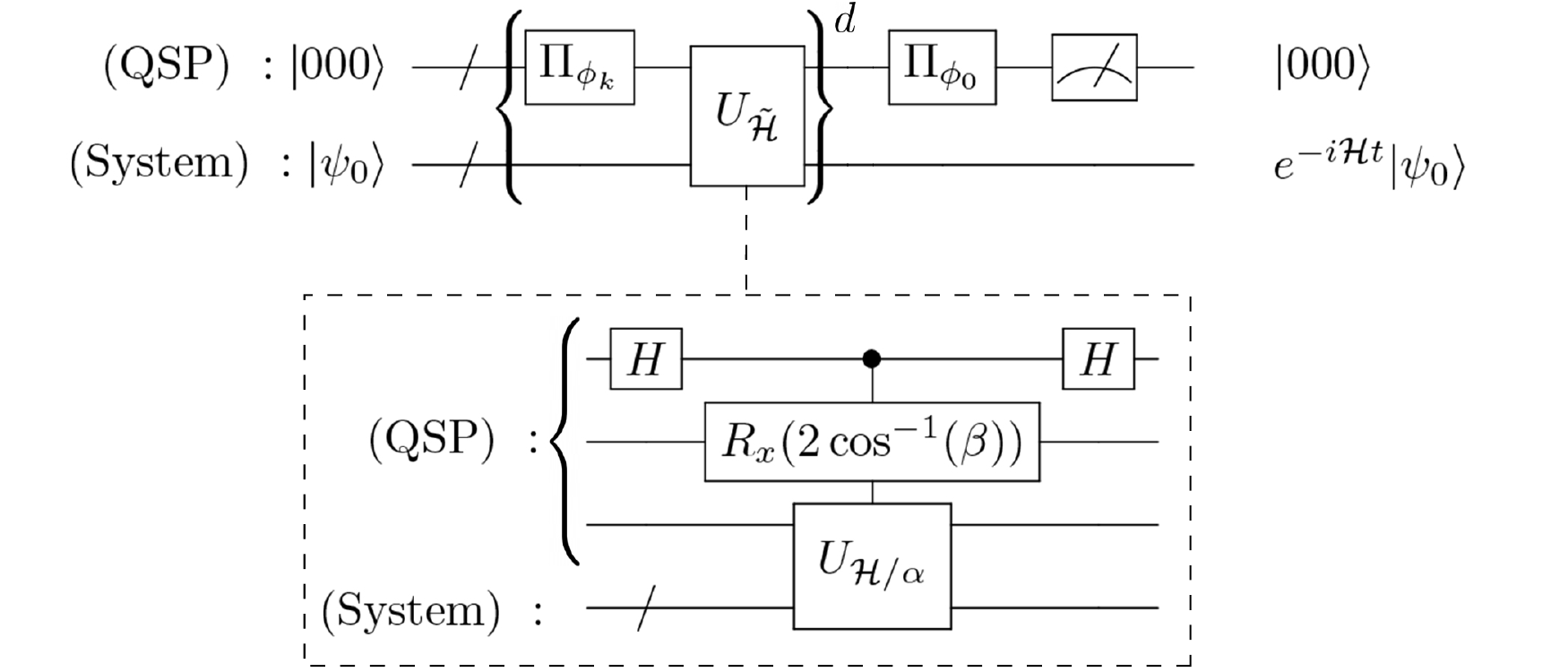}
    \end{center}
    \caption{The quantum circuit used to perform coherent one-shot simulation. Here, our projector is $\Pi = |000\rangle\langle 000|$ with corresponding projector-controlled phase shift is $\Pi_{\phi_k}$. The gates within brackets are repeated $d$ times for $k = d, d-1, ..., 1$ with phases $\vec{\phi}$ to construct the appropriate QET sequence for the polynomial $P_{\epsilon,(1-\beta)}^{\text{EECE}}(\tilde{\mathcal{H}};2t\alpha/\beta)$. This procedure outputs the correct time evolved state $e^{-i\mathcal{H}t}|\psi_0\rangle$ in the system register if the QSP qubits are measured in the state $|000\rangle$, which occurs with probability at least $1-2\epsilon$. We also note that the unitary $U_{\tilde{\mathcal{H}}}$ is defined in the inset as the block encoding of $\tilde{\mathcal{H}}$, which is attained by combining the circuits of Figs.~\ref{fig:HamScalingEncoding} and~\ref{fig:HamSumEncoding}.}
\label{fig:OneShotCircuit}
\end{figure*}

\subsection{Query Complexity}\label{Sec:Coherent_Runtime}
Let us now investigate the query complexity of coherent one-shot Hamiltonian simulation, for which we will need to determine the degree of the polynomial approximation to the EECE. As we mentioned earlier, we should choose $\Delta/2 \leq (1-\beta)/2$ such that $P_{\epsilon,\Delta}^{\text{EECE}}(x;\tau)$ acts as the complex exponential on the eigenvalues of interest, and $\tau = 2t\alpha/\beta$ such that we simulate for the correct time. Appropriately, our polynomial of interest is $P_{\epsilon,(1-\beta)}^{\text{EECE}}(x;\ 2t\alpha/\beta)$ which according to Eq.~(\ref{eq:OneShotPoly}) has degree $d = \text{deg}(P^{\text{sin}}_{\epsilon/6}(x;\ 2t\alpha/\beta)) +  \text{deg}(P^{\text{sign}}_{\epsilon/3,(1-\beta)}(x))$. 

Employing the bounds on the degrees of $P^{\text{sin}}_{\epsilon/6}(x;\ 2t\alpha/\beta)$ and $P^{\text{sign}}_{\epsilon/3,(1-\beta)}(x)$, which are discussed in Sec.~\ref{Sec:Conventional_Complexity} and Appendix~\ref{Sec:SignFunction}, respectively, we find that a sufficient degree to construct $P_{\epsilon,\Delta}^{\text{EECE}}(x;\ 2t\alpha/\beta)$ is
\begin{equation} \label{eq:NCo}
\begin{split}
    &2K \left(\frac{2t\alpha}{\beta}, \frac{\epsilon}{6} \right) + 1 + \gamma\left( \frac{\epsilon}{3},1-\beta \right) \\
    & \qquad =2 \cdot \Bigg\lfloor \frac{1}{2}r\left(\frac{e \alpha |t|}{\beta}, \frac{5\epsilon}{24} \right)\Bigg\rfloor + \gamma\left( \frac{\epsilon}{3},1-\beta \right) + 1\\
    & =: N^{\text{OS}}_{\mathcal{H}}(\epsilon, \delta=2\epsilon, \beta; t, \alpha),
\end{split}
\end{equation}
where $\gamma(\epsilon,\Delta) $ is a function explicitly defined in Eq.~(\ref{eq:gamma}). Because the transformed Hamiltonian to which QET is applied is linear in the Hamiltonian, this expression is the query complexity of our coherent one-shot Hamiltonian simulation, which we denote by $N_{\mathcal{H}}^{\text{OS}}(\epsilon, \delta=2\epsilon, \beta; t,\alpha)$ in the last line of the above equation (with superscript ``OS" for ``one-shot"). This expression implicitly bounds the probability of failure as $\delta \leq 2\epsilon$ as per Sec.~\ref{Sec:Coherent_Success}, which we note in the query complexity. 

Using the appropriate asymptotic scalings, the $r(\cdot,\cdot)$ term in $N^{\text{OS}}_{\mathcal{H}}(\epsilon, \delta=2\epsilon, \beta; t, \alpha)$ scales as $\sim 1/\beta$, whereas the $\gamma(\cdot, \cdot)$ term scales as $\sim 1/(1-\beta)$. As these terms are added together in $N^{\text{OS}}_{\mathcal{H}}(\epsilon, \delta=2\epsilon, \beta; t, \alpha)$, there is a clear trade off between choosing a large or small value for $\beta$, the optimal value of which depends on the details of the Hamiltonian (through $\alpha)$ and the desired error rate (through $\epsilon)$. Nonetheless, $\beta$ may be chosen as some fixed constant (say, $\beta = 1/2$), and hence the query complexity of coherent one-shot Hamiltonian simulation scales as 
\begin{equation}
    \begin{split}
        N_{\mathcal{H}}^{\text{OS}} = \Theta\big(\alpha |t| + \ln(\tfrac{1}{\epsilon}) + \ln(\tfrac{1}{\delta})\big).
    \end{split}
    \label{query-coh}
\end{equation}
This certainly provides an improvement in the scaling with respect to $\delta$ over the QSP-LCU+AA coherent simulation algorithm. We have thus demonstrated a central result of this paper:
\begin{theorem}[Coherent One-Shot Hamiltonian Simulation] \label{thm:coherent-one-shot}
Provided a block encoding of the Hamiltonian (possibly rescaled), the Coherent One-Shot Hamiltonian Simulation algorithm achieves efficient fully-coherent Hamiltonian simulation with arbitrarily small $\epsilon$ and $\delta$, while querying (the block encoding of) the Hamiltonian a total number of times 
\begin{equation}
    \Theta\left( \alpha |t| + \ln(\frac{1}{\epsilon}) + \ln(\frac{1}{\delta}) \right).
\end{equation}
\end{theorem}

\section{Comparison on the Scaling of Query Complexity}\label{Sec:Scaling}

Given the analytical bounds on query complexities of the QSP-LCU+AA and coherent one-shot simulation algorithms derived in Eqs.~(\ref{eq:NAmp}),~(\ref{eq:N_OAA}), and~(\ref{eq:NCo}), we now study the behavior of the query complexity with respect to simulation time and error. We compare and contrast the scaling of the complexity between different algorithms discussed above. We note that a similar analysis is presented in Ref.~\cite{childs2018toward}, wherein the query complexities of various Hamiltonian simulation algorithms are compared with each other and with empirical bounds, although coherence is not a focal point of their work.

Let us begin by investigating the scaling of the query complexity with time. With parameter choices $\alpha=5$, $\beta = 0.5$, and fixed $\epsilon =  0.02$ we calculate $N_{\mathcal{H}}^{\mathrm{AA}}$, $N_{\mathcal{H}}^{\mathrm{ROAA}}$, and $N_{\mathcal{H}}^{\mathrm{OS}}$ for a variety of $\epsilon$ using Eqs.~(\ref{eq:NAmp}) and~(\ref{eq:NCo}), respectively. Because the coherent one-shot algorithm naturally has failure probability bounded above by $2\epsilon$, we set $\delta=2\epsilon$ in the computation of $N_{\mathcal{H}}^{\mathrm{AA}}$, which provides a fair comparison between the two algorithms. We plot these results in the left panel of Fig.~\ref{fig:QueryComplexity}, showcasing all of the query complexities in the main figure and zooming into $N_{\mathcal{H}}^{\mathrm{OS}}$ and $N_{\mathcal{H}}^{\mathrm{ROAA}}$ in the inset. As expected, these query complexities all scale linearly with time, with the QSP-LCU+AA method growing the most rapidly. The inset indicates that $N_{\mathcal{H}}^{\mathrm{OS}}$ has a shallower slope than $N_{\mathcal{H}}^{\mathrm{ROAA}}$, suggesting the superiority of the coherent one-shot algorithm for long time simulations. 

Next, we look at the scaling of query complexity with error, $\epsilon$. Again, with parameter choices $\beta = 0.5$, $\alpha=5$, and fixed $t=5$, we calculate $N_{\mathcal{H}}^{\mathrm{AA}}$, and $N_{\mathcal{H}}^{\mathrm{OS}}$, and $N_{\mathcal{H}}^{\mathrm{ROAA}}$ for a variety of errors $\epsilon$, and set $\delta=2\epsilon$ to provide an accurate comparison between the algorithms. As before, we plot the results in the right panel of Fig.~\ref{fig:QueryComplexity}. We see that the QSP-LCU+AA algorithm attains the greatest query complexity with the most severe scaling. Moreover, the inset shows that $N_{\mathcal{H}}^{\mathrm{OS}}$ exceeds $N_{\mathcal{H}}^{\mathrm{ROAA}}$ for small $\epsilon$ indicating the utility of oblivious amplitude amplification for achieving highly accurate simulations. 

In summary, both our coherent one-shot algorithm and the QSP-LCU+ROAA algorithms achieve significant speedup as compared to QSP-LCU+AA. However, there is a performance trade-off depending on the parameter regimes of simulation error $\epsilon$ and simulation time $t$ as evident from the crossing in the insets of Fig. \ref{fig:QueryComplexity}. In the case of $\delta = 2\epsilon$, our coherent one-shot algorithm out-performs QSP-LCU+ROAA for long simulation time with a moderate error tolerance.

\begin{figure*}[htbp]
  \includegraphics[width=15cm,height=7.5cm]{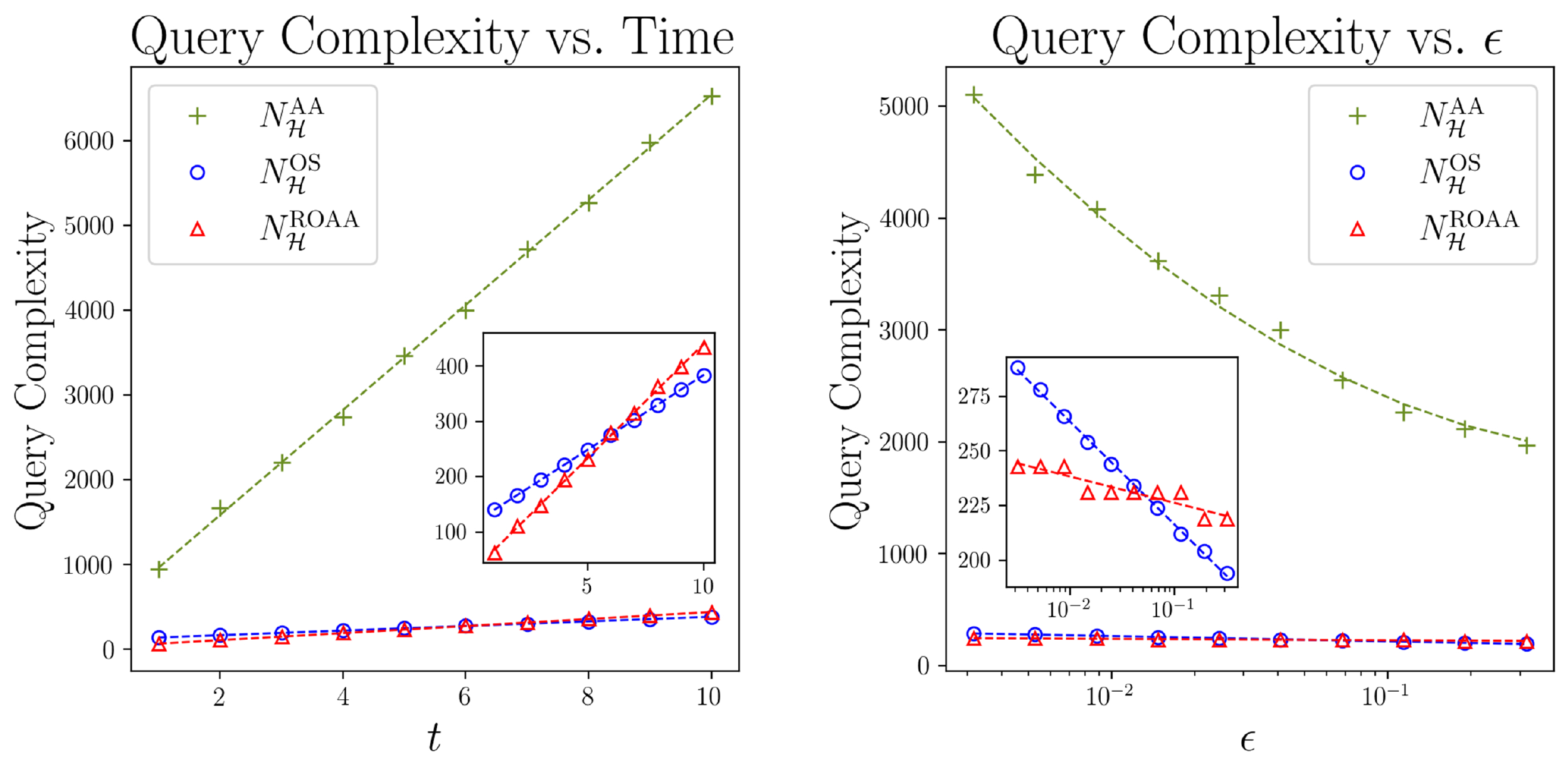}
  \caption{Query complexity vs. $t$ and $\epsilon$ for the coherent Hamiltonian simulation algorithms presented in this paper ($\text{N}_{\mathcal{H}}^{\text{AA}}$ for QSP-LCU+AA, $N_{\mathcal{H}}^{\mathrm{ROAA}}$ for robust oblivious amplitude amplification, and $\text{N}_{\mathcal{H}}^{\text{OS}}$ for coherent one-shot simulation), as well as their lines of best fit. The insets zoom into $\text{N}_{\mathcal{H}}^{\text{OS}}$ and $\text{N}_{\mathcal{H}}^{\text{ROAA}}$ for a more clear comparison.}
  \label{fig:QueryComplexity}
\end{figure*}

\section{Applications and Simulation Results}
\label{Sec:Applications}

One of the most promising applications of quantum algorithms is simulating the dynamics of interacting quantum many-body systems, which is known to be exponentially hard on classical computers due to particle-particle interactions. In this section, we employ the aforementioned fully-coherent simulation algorithms to simulate the Heisenberg model in both constant and time-dependent external fields (Sec. \ref{Sec:heisenberg}), as well as the electronic charge migration dynamics in a hydrogen molecule (Sec. \ref{sec:electronic-dynamics}). While our numerical simulations are limited to two spins in the Heisenberg model and minimal basis for the hydrogen molecule, they suffice to demonstrate the strengths of these fully-coherent algorithms. These algorithms can be readily applied to systems of larger size and to more complex molecules, provided an efficient block-encoding of the system Hamiltonian. We restrict our numerical simulations to the QSP-LCU+ROAA and the coherent one-shot algorithm, as these techniques have exhibited the greatest efficiencies in terms of query complexity.

\subsection{The Heisenberg Model}\label{Sec:heisenberg}

We use the QSP-LCU+ROAA and coherent one-shot Hamiltonian simulation algorithms to simulate a two-spin Heisenberg model with nearest neighbor interactions and an external magnetic field on each site. We discuss the full details of this model in Sec.~\ref{sec:heisenberg-ham}. The simulation results are presented in Sec.~\ref{sec:heisenberg-dyn}, which further illustrate the advantages of the full-coherence property in Hamiltonian simulation.

\subsubsection{The Hamiltonian} \label{sec:heisenberg-ham}

We focus on the one-dimensional Heisenberg spin chain with nearest-neighbor interaction, subjected to a site-specific time-dependent magnetic field along the $z$-direction, whose Hamiltonian is
\begin{align}
    \mathcal{H}(t) &= \mathcal{H}_0(t) + \mathcal{H}_1, \label{eq:heisenberg-Hfull}\\
    \mathcal{H}_0(t) &= \sum_{j = 1}^n h_j(t) \sigma_j^z, \\
    \mathcal{H}_1 &= \sum_{j = 1}^n ( g_j^x \sigma_j^x \sigma_{j+1}^x + g_j^y \sigma_j^y \sigma_{j+1}^y + g_j^z \sigma_j^z \sigma_{j+1}^z ),
    \label{eq:heisenberg-H1}
\end{align}
where $n$ is the total number of spins, $\{ g_j^x, g_j^y, g_j^z \}$ are the interaction strengths along $x, y, z$ directions between spins $j$ and $j+1$, and $h_j(t)$ is the on-site (time-dependent) magnetic field for spins $j = 1, \ldots, n$. At a high level, $\mathcal{H}_0(t)$ and $\mathcal{H}_1$ are the non-interacting term for each spin in the magnetic field and the pairwise nearest-neighbor interaction term, respectively. %$n+1 = 1$ for the periodic boundary condition, . $\mathcal{H}_0(t)$ is the non-interacting term for each spin in the magnetic field, and $\mathcal{H}_1$ is the pairwise nearest-neighbor interaction term. In our benchmark, we choose a model with $n=2$ spins.

We further choose a system size of $n=2$. Although the simulation algorithms discussed in this paper apply equally well to larger systems as long as one is provided a block encoded Hamiltonian, simulations of larger systems quickly become difficult to emulate on classical hardware. The two spin Heisenberg Hamiltonian components $\mathcal{H}_0$ and $\mathcal{H}_1$ can be written explicitly in the computational basis $\{ \ket{00}, \ket{01}, \ket{10}, \ket{11} \}$ of $\sigma_z$ eigenvectors as
\begin{align}
    \mathcal{H}_0 &= \text{diag}\{ h_+(t), h_-(t), - h_-(t), - h_+(t) \}, \nonumber \\
    \mathcal{H}_1 &= 
    \begin{bmatrix}
        g_z &&& g_x - g_y \\
        & -g_z & g_x + g_y& \\
        &g_x + g_y& -g_z & \\
        g_x - g_y&&& g_z
    \end{bmatrix},
    \label{2spin-h}
\end{align}
where $h_+(t) = h_1(t) + h_2(t)$, $h_-(t) = h_1(t) - h_2(t)$, and $\ket{0}$ and $\ket{1}$ represent spin up and down, respectively. It can be shown that the commutator between two instances of the Hamiltonian evaluated at different times is
\begin{equation} \label{eq:commutator}
\begin{split}
    &[\mathcal{H}(t_1), \mathcal{H}(t_2)] =\\
    & \ \  2i \big( (g_y \Delta h_2 - g_x \Delta h_1) \sigma_1^y \sigma^x_2 + (g_y\Delta h_1 - g_x \Delta h_2) \sigma^x_1 \sigma^y_2 \big),
\end{split}
\end{equation}
where $\Delta h_j = h_j (t_2) - h_j (t_1)$ for $j = 1, 2$. Therefore, in general $[\mathcal{H}(t_1), \mathcal{H}(t_2)] \ne 0$ if $t_2 \ne t_1$, so instances of the Hamiltonian corresponding to different times do not necessarily commute.

\subsubsection{Simulating the Spin Dynamics} \label{sec:heisenberg-dyn}

To illustrate the advantages of full-coherence in Hamiltonian simulation, we use the QSP-LCU+ROAA and coherent one-shot methods to simulate the spin dynamics of the Heisenberg Hamiltonian given in Eq.~\eqref{2spin-h} for both time-independent \emph{and} time-dependent magnetic fields. The data and scripts used in our simulation are documented in Ref. \cite{mlcc2021repo}.

After running both algorithms on a given initial state $\ket{\psi_0}$ with a fixed number of queries to the Hamiltonian block-encoding, we measure the expectation values of the first spin's $z$ component, $\langle \sigma^z_1 \rangle$, for various simulation time intervals, and compare them to analytic results obtained by applying the exact time evolution operator to $\ket{\psi_0}$. In the following subsections, we provide details on the block-encoding used and then explore the numerical results for the time-independent and time-dependent Heisenberg Hamiltonian simulations.

\subsubsubsection{Block-Encoding} \label{sec:numsim-be}

As a necessary preliminary to both algorithms, we first construct unitary matrices which block-encode the total two-spin Hamiltonian. For the QSP-LCU+ROAA algorithm, we require a block encoding $U_{\mathcal{H}/\alpha}$ of $\mathcal{H}/\alpha$, for which it suffices to choose:
\begin{align}
    U_{\mathcal{H}/\alpha} = 
    \begin{bmatrix}
        \mathcal{H} / \alpha & \sqrt{I - \mathcal{H}^2/\alpha^2} \\
        \sqrt{I - \mathcal{H}^2/\alpha^2} & - \mathcal{H} / \alpha
    \end{bmatrix}
    \label{eq:be-2spin}
\end{align}
where $\mathcal{H} = \mathcal{H}_0 + \mathcal{H}_1$ given by Eq.~\eqref{2spin-h} and the rescaling factor $\alpha \geq \| \mathcal{H} \|$ ensures that $\sqrt{I - \mathcal{H}^2/\alpha^2}$ is real. We illustrate a circuit based implementation of $U_{\mathcal{H}/\alpha}$ in Appendix~\ref{sec:unitaryBlockEncodingCirc}.

For the coherent one-shot algorithm, we employ the pre-transformation shown in the circuit of Fig.~\ref{fig:OneShotCircuit} to construct a block encoding $U_{\tilde{\mathcal{H}}}$ of $\tilde{\mathcal{H}} = \frac{1}{2}(I+\beta \mathcal{H}/\alpha)$. Then, the quantum circuits of Fig. \ref{fig:ObliviousAmplitudeAmpSimulationCircuit} is used to simulate the time evolution with the QSP-LCU+ROAA method, while the circuit in Fig.~\ref{fig:OneShotCircuit} is employed to simulate via the coherent one-shot algorithm.

\subsubsubsection{Time-Independent Simulation} \label{sec:numsim-t-ind}

We begin by using both of the efficient fully-coherent algorithms to simulate the spin dynamics of the aforementioned Heisenberg model with a constant external field, while making a fixed number of queries to the Hamiltonian. 

The initial state is chosen to be $\ket{\psi_0} = \ket{00}$, indicating that both spins are initially spin-up. We then choose a value for the simulation time interval $t$ and, using the Qiskit Aer quantum circuit simulator package \cite{Qiskit}, apply the quantum circuits for each algorithm
%displayed in Figs. \ref{fig:AmplitudeAmpSimulationCircuit} and \ref{fig:OneShotCircuit} 
to a copy of the state $\ket{\psi_0}$ to obtain the time-evolved states $\ket{\psi_\textrm{ROAA}} = \left[e^{-i\mathcal{H}t}\right]_{\textrm{ROAA}}\ket{\psi_0}$ and $\ket{\psi_\textrm{OS}} = \left[e^{-i\mathcal{H}t}\right]_{\textrm{OS}}\ket{\psi_0}$, which correspond to the outputs of the QSP-LCU+ROAA and coherent one-shot algorithms for simulation time $t$, respectively. Finally, for each time-evolved state $\ket{\psi}\in\{\ket{\psi_\textrm{ROAA}},\ket{\psi_\textrm{OS}}\}$ we compute the expectation value of the first spin's $z$ component with
\begin{equation}\label{eq:spinExpectation}
    \langle \sigma_1^z \rangle = \bra{\psi}(\sigma^z \otimes I)\ket{\psi}.
\end{equation}
This expectation value may be evaluated readily on a quantum computer by repeatedly preparing the state $\ket{\psi}$ and measuring the first qubit in the computational basis (which is the $\sigma^z$ eigenbasis). Suppose a total of $n=n_0+n_1$ such measurements are performed on $n$ copies of $\ket{\psi}$, such that $n_0$ measurements observe the first qubit in the state $\ket{0}$, and $n_1$ measurements in the state $\ket{1}$. Then the expectation value $\langle \sigma_1^z\rangle$ is approximately given by 
\begin{align}
    \langle \sigma_1^z\rangle \approx \frac{1}{n}(n_0 - n_1).
\end{align}
In any practical implementation of these algorithms, the number of measurements $n$ can be chosen large enough such that the measured $\langle \sigma_1^z\rangle$ converges satisfactorily to its true value. This entire process of preparing the algorithmically time-evolved state and inferring the expectation value $\langle \sigma_1^z\rangle$ is repeated for all simulation times $t$ over the time domain of interest.

To provide a standard to benchmark the results of our QSP simulation algorithms, we determine the exact time-evolution operator by solving for the eigenvalues $\lambda$ of the Hamiltonian (with associated eigenvectors $\ket{\lambda}$) and explicitly computing the exact time-evolution operator as
\begin{equation}\label{eq:exact_evolution}
    \left[e^{-i\mathcal{H}t}\right]_{\textrm{Exact}} = \sum_\lambda e^{-i \lambda t}\ket{\lambda}\bra{\lambda}.
\end{equation}
The exact values for $\langle \sigma_1^z \rangle$ can be calculated by substituting $\ket{\psi} =  \left[e^{-i\mathcal{H}t}\right]_{\textrm{Exact}}\ket{\psi_0}$ into Eq.~(\ref{eq:spinExpectation}).

As an example, we first set $h_1(t) = h_2(t) = 0.5$ for all times $t$ and $g_x = 1$, $g_y = g_z = 0$. The spin expectation $\langle \sigma_1^z \rangle$ and absolute error $\abs{\langle \sigma_1^z \rangle - \langle \sigma_1^z \rangle_\textrm{Exact}}$ were measured using both algorithms over the time domain $t \in [0,3.5]$.

\begin{figure*}[htbp]
  \includegraphics[width=14cm,height=7cm]{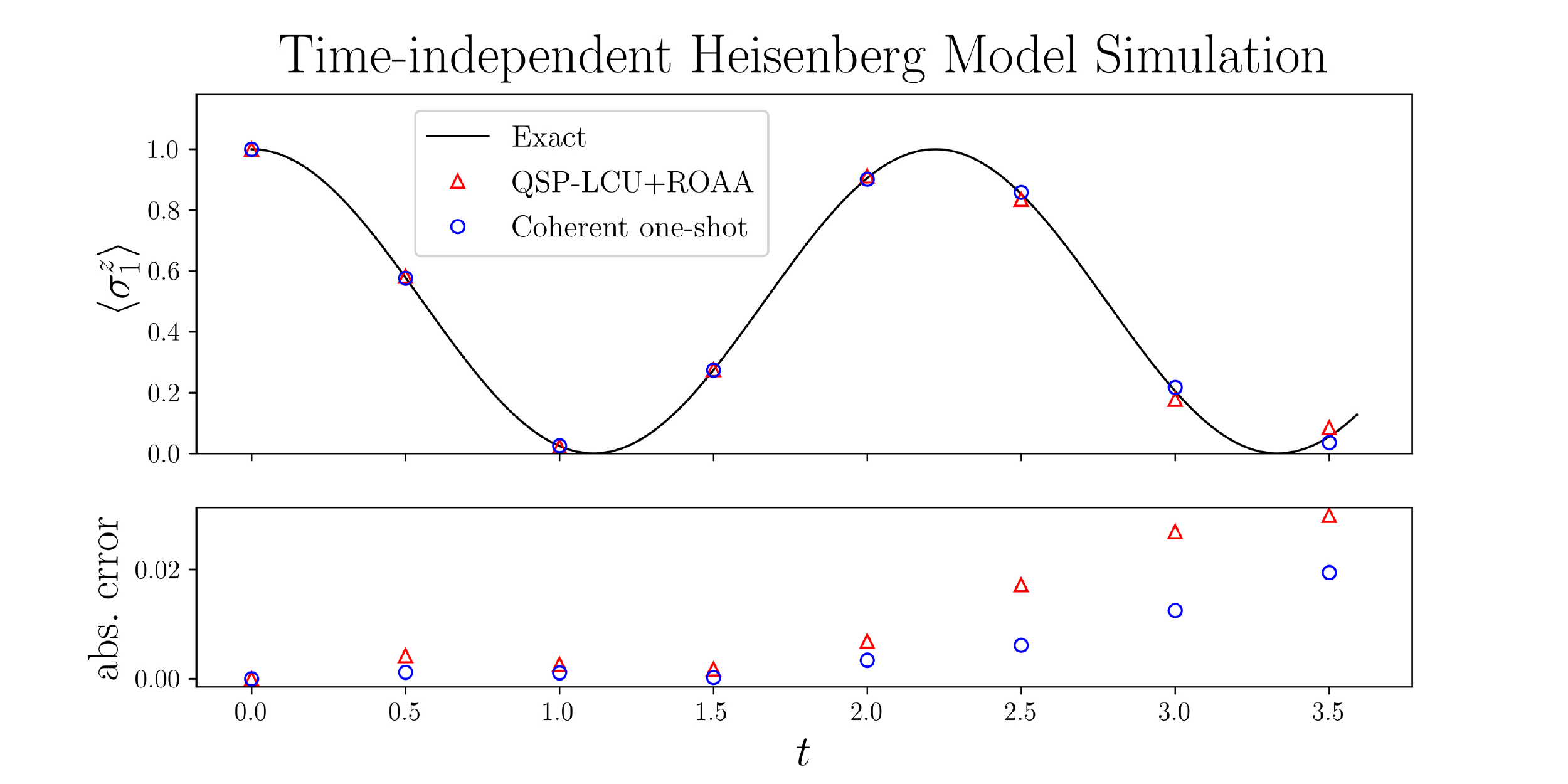}
  \caption{Hamiltonian simulation for a two-spin Heisenberg model with a constant magnetic field $h_1 = h_2 = 0.5$ using both the conventional LCU algorithm with robust oblivious amplitude amplification (QSP-LCU+ROAA) and the coherent one-shot algorithm. The QSP-LCU+ROAA scheme made 33 queries to the Hamiltonian block-encoding for each simulation interval, while the one-shot scheme made 32. \textbf{(Top)} Expectation value of the first spin's $z$ component at each simulation time interval. \textbf{(Bottom)} Absolute error between algorithm output and exact result at each simulation time interval. } 
  \label{fig:timeIndependent}
\end{figure*}

The degrees of the QSP polynomials used for the two simulation schemes were chosen so that the number of queries made to the Hamiltonian block encoding in each algorithm would be nearly equivalent. For the QSP-LCU+ROAA algorithm, polynomials with degrees $d_{\textrm{cos}} = 6$ and $d_{\textrm{sin}} = 5$ were used to approximate $\cos(\mathcal{H}t)$ and  $-i\sin(\mathcal{H}t)$, respectively. The total number of queries made to the Hamiltonian block encoding for the entire QSP-LCU+AA algorithm is thus $N^{\textrm{ROAA}}_\mathcal{H} = 3(d_{\textrm{cos}}+d_{\textrm{sin}}) = 33$. Note that $N^{\textrm{ROAA}}_\mathcal{H}$ is necessarily odd here because $d_{\textrm{sin}}$ must both be odd while $d_{\textrm{cos}}$ must be even. On the other hand, for the coherent one-shot algorithm, the number of queries $N^{\textrm{OS}}_\mathcal{H}$ is necessarily even since the EECE is approximated with an even degree polynomial; subsequently, the EECE was approximated with degree $d_\textrm{EECE} = 32$ so that $N^{\textrm{OS}}_\mathcal{H} = 32$ was nearly equal to $N^{\textrm{AA}}_\mathcal{H}$. 

The QSP phases corresponding to these polynomials were generated using the PyQSP optimization package with the scaling parameters $\alpha = 1.5$ and $\beta = 0.4$. Again, the parameter $\alpha$ was chosen to be slightly larger than the largest eigenvalue of $\mathcal{H}$ (which is equal to $\sqrt{2}$ in this case) so that all blocks of the Hamiltonian block encoding are real. Regarding the parameter $\beta$, recall that by Eq.~(\ref{eq:NCo}) we expect an optimal $\beta$ to exist for a given Hamiltonian and the desired error probability. However, the exact value of the optimal $\beta$ also depends on the choice of QSP polynomial, and therefore may be difficult to calculate directly. Recognizing nonetheless that very long effective simulation times $\tau \propto 1/\beta$ must lead to increased error (and therefore decreased success probability), we selected a sufficiently large $\beta$ to reduce the effective simulation time and boost the success probability. Note that $\beta$ cannot be chosen too large, since the polynomial approximation to the EECE over the range $[\frac{1-\beta}{2}, \frac{1+\beta}{2}]$ will become inaccurate if this range is too wide.

The simulation results for the time-independent Heisenberg model are plotted in Fig. \ref{fig:timeIndependent}. While both algorithms very closely track the exact result for $t \leq 1.5$, the outputs of both algorithms begin to gradually accrue error at later time points. In this particular simulation, we observe that our novel coherent one-shot method introduces less error than QSP-LCU+ROAA at these later time points when each algorithm is provided (nearly) the same number of Hamiltonian queries; this observation is corroborated by the results of Fig. \ref{fig:QueryComplexity}, which suggest that, in the large-time limit, one-shot should require fewer queries than QSP-LCU+ROAA to achieve a given simulation error. 

%the coherent one-shot algorithm outperforms the QSP-LCU+AA algorithm over this time interval, attaining a lower average absolute error of $0.0081$ as compared to an average error of $0.0137$ suffered by the QSP-LCU+AA algorithm. However, for times $t>2.5$ the QSP-LCU+AA output diverges from the exact result and accrues large error, while the coherent one-shot output continues to agree well with the exact dynamics. Evidently, the coherent one-shot algorithm is significantly more accurate than the QSP-LCU+AA algorithm, especially at large times.

%The failure of the QSP-LCU+AA approach at large time intervals can be attributed to the low degree polynomials used to approximate $\cos(\mathcal{H}t)$ and $\sin(\mathcal{H}t)$ ($d_{\textrm{cos}} = 2$ and $d_{\textrm{sin}} = 3$, respectively), which become inaccurate for large $t$. These degrees can only be increased at the expense of decreasing the degree of the polynomial approximation to the sign function, which will negatively affect the success probability. In contrast, in the coherent one-shot approach, $\operatorname{EECE}(x;\ t)$ is approximated more satisfactorily at large $t$ by a 24-degree polynomial. For a fixed budget of Hamiltonian queries, all of the queries made in the one-shot approach are used to improve the polynomial approximation to the complex exponential, whereas many of the queries made in the QSP-LCU+AA approach are instead used to perform the costly amplitude amplification which is necessary for preserving coherence. 

Lastly, the success probability ($1-\delta$) of each algorithm is determined by the probability that the desired block containing the $e^{-i \mathcal{H}t}$ matrix is accessed upon measurement of the circuit. The success probabilities for the QSP-LCU+ROAA and coherent one-shot algorithms were calculated to be $99.8\%$ and $97.4\%$ on average over the range of time $t\in[0,3.5]$, respectively. These near-unity success probabilities indicate that the two methods indeed maintain a high degree of coherence.

\subsubsubsection{Time-Dependent Simulation} \label{sec:time-dependent}

We now explore the application of both efficient fully-coherent simulation algorithms to a Heisenberg model with a time-dependent magnetic field. In general, the time-evolution operator $\mathcal{U}_\mathcal{H}(t)$ of a time-dependent Hamiltonian $\mathcal{H}(t)$ corresponding to an evolution time $t$ is given by
\begin{equation}\label{eq:timeEvolutionOperator}
    \mathcal{U}_\mathcal{H}(t) = \mathcal{T}\exp[-i \int_0^t \mathcal{H}(t')dt']
\end{equation}
where $\mathcal{T}$ is the time-ordering operator. This expression indicates that if the Hamiltonian does not commute with itself at different times, it does not suffice to naively integrate the Hamiltonian over the desired simulation period and exponentiate to obtain the time-evolution operator. Hence, time-dependent simulation cannot be performed with only a single call to either fully-coherent simulation algorithm, a byproduct of the Hamiltonian's time-dependence and non-commutativity (as in Eq.~(\ref{eq:commutator})).

To circumvent this difficulty, we Trotterize the full time-evolution operator into the product of many short time-evolution operators where the Hamiltonian is assumed to be constant over each short time period:
\begin{equation}\label{eq:approximateTimeEvolution}
    \mathcal{U}_\mathcal{H}(t) \approx \prod_{k=0}^{L-1}\exp[-i\mathcal{H}(k\Delta t) \Delta t].
\end{equation}
Here, $L$ is the number of discrete (Trotter) time steps used in the decomposition, and $\Delta t = t/L$ is the time step size. In the limit that $L\rightarrow\infty$ and $\Delta t\rightarrow 0$, the product approaches the exact time-evolution operator. 

To implement Trotterization with our two algorithms, we employ each of the QSP-LCU+ROAA and coherent one-shot techniques to approximate the complex exponentials $\exp[-i\mathcal{H}(k \Delta t)\Delta t]$ for $k = 0,1\ldots, L-1$, resulting in a length $L$ sequence for each algorithm. To simulate the time-evolution, the operators of each sequence are applied successively to a given input state $\ket{\psi_0}$ in accordance with Eq.~(\ref{eq:approximateTimeEvolution}), and the final output state is a close approximation to $\mathcal{U}_\mathcal{H}(t)\ket{\psi_0}$ for sufficiently small $\Delta t$.

Note that an implementation of Trotterization using an incoherent simulation algorithm (with a non-unit success probability) will require post-selection after each time slice. For the QSP-LCU simulation algorithm of Sec.~\ref{Sec:ConventionalHamSim}, which has probability of success close to $1/4$, this implies that the overall success probability decreases as $(1-\delta) \sim  4^{-L}$, which approaches zero exponentially fast with an increasing number of Trotter steps $L$. This ruins the quantum simulation quickly and essentially forbids long-time Hamiltonian simulations. In contrast, we expect that our coherent Hamiltonian simulation protocol provides significant speed up in this case, because the simulation of $L$ time slices may be chained together coherently via simple multiplication without an exponential increase in the probability of failure.

Also note the fundamental difference between our strategies for simulating the time-independent and time-dependent cases: for the time-independent case, the output at each time $t$ was the result of a single call to one of the fully-coherent simulation algorithms, and a new polynomial approximation to the complex exponential $e^{-ix t}$ was generated for each simulation time $t$. In contrast, for the time-dependent case, the output at each time $t = k\Delta t$ is the result of $k$ calls to one of our coherent simulation algorithms applied successively to the input state, and the \textit{same} polynomial approximation to $e^{-ix \Delta t}$ is used for each call. The disadvantage of the latter strategy is that the total query complexity of the simulation acquires a multiplicative prefactor equal to the number of Trotter steps $L$, as is analyzed in the following. 

The query complexity of the single call to coherent one-shot algorithm is given by Eq. \eqref{query-coh}, such that the query complexity of $L$ calls is given by
\begin{align}
    N_{\mathcal{H}}^{\rm OS, L} = \Theta \bigg( L \cdot \left[ \alpha |\Delta t| + \ln(\frac{1}{\epsilon'}) + \ln(\frac{1}{\delta'}) \right] \bigg),
    \label{query-coh-trotter1}
\end{align}
where $\epsilon'$ and $\delta'$ are the error and failure probability for each call of time $\Delta t$. To guarantee a final error of $\epsilon$ and final failure probability $\delta$ after $L$ calls, it is sufficient to select $\epsilon' < \epsilon / L$ and $\delta' < \delta / L$, which upon substitution into Eq.~\eqref{query-coh-trotter1} yields
\begin{align}
    N_{\mathcal{H}}^{\rm OS, L} = \Theta \bigg( \alpha |t| + L \ln(\frac{L}{\epsilon}) + L \ln(\frac{L}{\delta}) \bigg).
    \label{query-coh-trotter2}
\end{align}
It is easy to see that $N_{\mathcal{H}}^{\rm OS, L} \ge N_{\mathcal{H}}^{\rm OS}$ by comparison with Eq. \eqref{query-coh}, with equality at $L = 1$, as is expected.

As a benchmark for our simulations of the time-dependent model, we first calculate the exact time-evolution operator $\mathcal{U}_\mathcal{H}(t)$ over the desired domain of $t$. We begin with the time-dependent Schr\"odinger equation for the evolution of the state $\ket{\psi(t)}$:
\begin{equation}
    i \frac{d}{dt}\ket{\psi(t)} = \mathcal{H}(t)\ket{\psi(t)}.
\end{equation}
Substituting $\ket{\psi(t)} = \mathcal{U}_\mathcal{H}(t)\ket{\psi_0}$ into the above equation and noting that $\ket{\psi_0}$ is constant, we obtain the following ordinary differential equation:
\begin{equation}\label{eq:timeEvolutionODE}
    i \frac{d}{dt}\mathcal{U}_\mathcal{H}(t) = \mathcal{H}(t)\mathcal{U}_\mathcal{H}(t).
\end{equation}
Observe that solutions $\mathcal{U}_\mathcal{H}(t)$ of this equation will be of the form given in Eq.~(\ref{eq:timeEvolutionOperator}). The matrix version of Eq. \eqref{eq:timeEvolutionODE} was solved numerically in Mathematica for finely-spaced discrete times $t\in[0,15]$, and the exact value for $\langle \sigma_1^z \rangle$ at a given time $t$ was determined by applying the corresponding exact solution of $\mathcal{U}_\mathcal{H}(t)$ to $\ket{\psi_0}$ and substituting the resulting state into Eq.~(\ref{eq:spinExpectation}).

\begin{figure*}[htbp]
  \includegraphics[width=14cm,height=7cm]{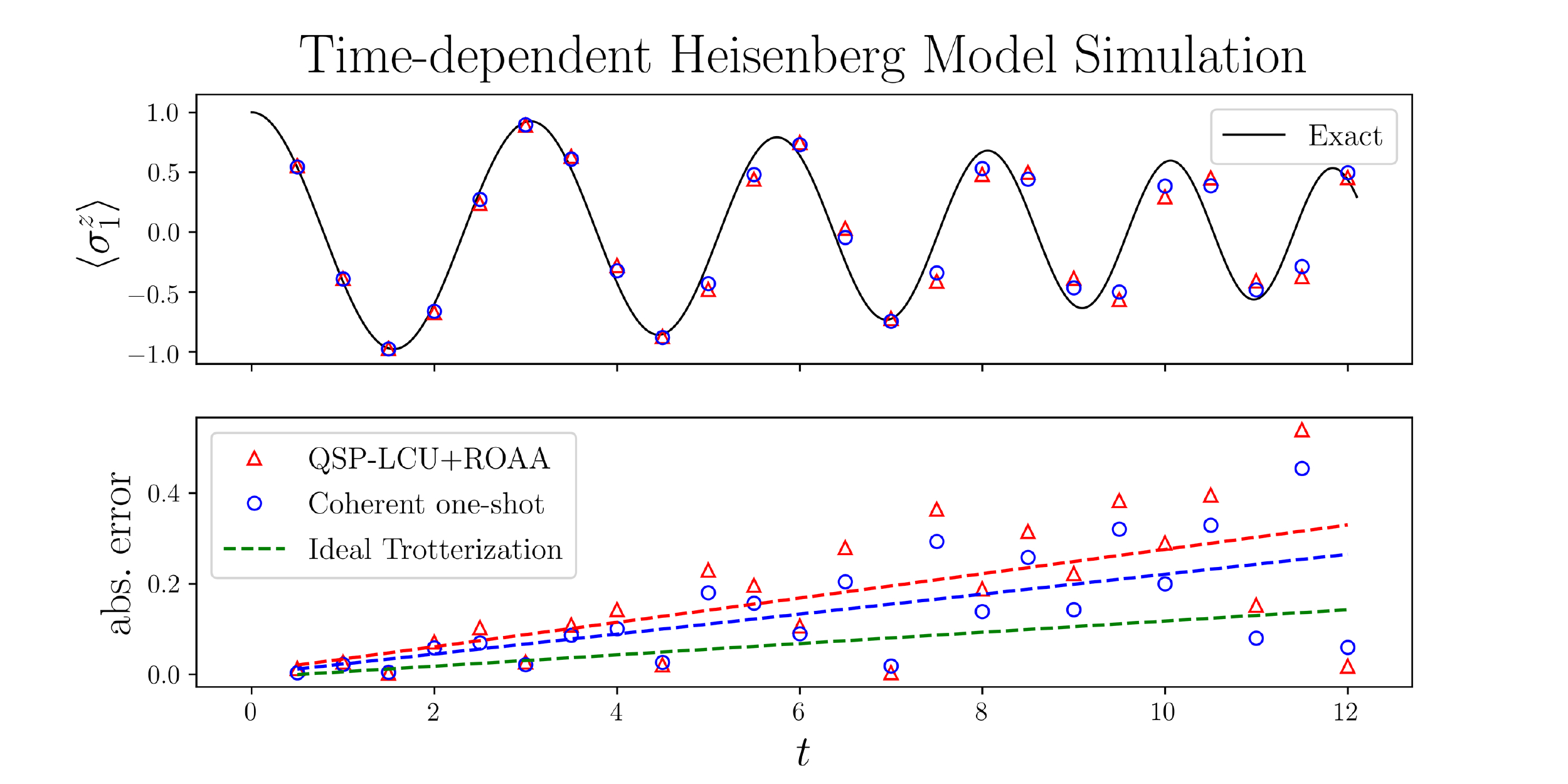}
  \caption{Hamiltonian simulation for a two-spin 1D Heisenberg model with a linearly increasing magnetic field $h_1(t) = h_2(t) = t/15$ using both the conventional LCU algorithm with robust oblivious amplitude amplification (QSP-LCU+ROAA) and the coherent one-shot algorithm. The QSP-LCU+ROAA scheme made 15 queries to the Hamiltonian block-encoding for each Trotter step, while the one-shot scheme made 14. A Trotter step size of $\Delta t = 0.5$ is used. \textbf{(Top)} Expectation value of the first spin's $z$ component $\langle \sigma^z_1 \rangle$ at each simulation time interval. \textbf{(Bottom)} Absolute error between algorithm output and exact result at each simulation time interval. Trendlines (appropriately colored) are also included. The (non-QSP) error intrinsic to the Trotterization was calculated at each step, and a linear trendline for this ideal Trotterization error is shown in green.} 
  \label{fig:timeDependent}
\end{figure*}

We simulated the evolution of the Heisenberg system for $t \in [0, 12]$ with the same initial state $\ket{\psi_0} = \ket{00}$ as in the time-independent case, but instead with a linearly increasing onsite magnetic field $h_1(t) = h_2(t) = t/15$. Note for this simulation that $g_x = 1$, $g_y = g_z = 0$ were chosen again. A Trotter step size of $\Delta t = 0.5$ was used in the simulation, corresponding to a total number of Trotter steps $L = 24$. At each step, in the QSP-LCU+ROAA scheme, we used degrees $d_{\text{cos}}=2$ and $d_{\text{sin}}=3$ to approximate $\cos(\mathcal{H}\Delta t)$ and $\sin(\mathcal{H}\Delta t)$, respectively; in the coherent one-shot scheme, we used a $d=14$ degree polynomial. Thus, the number of queries made to the Hamiltonian per step (or equivalently, per call to the simulation algorithm being used) was 15 for QSP-LCU+ROAA and 14 for coherent one-shot. The QSP phases corresponding to these polynomials were also solved using the PyQSP optimization package with the scaling parameters $\alpha = 2.5$ and $\beta = 0.25$; the QSP circuits for each algorithm were again modeled using the Qiskit Aer simulator. Similarly to the time-independent case, the parameter $\alpha$ was chosen to be greater than the largest eigenvalue of $\mathcal{H}(t)$ (which is $\sqrt{5}$ in this example), while the parameter $\beta$ was chosen to be sufficiently large such that a high success probability was still achieved for the one-shot scheme.

The simulation results for the time-dependent case are displayed in Fig. \ref{fig:timeDependent} with $\langle \sigma_1^z \rangle$ and its absolute error plotted at each Trotter step in the upper and lower panel, respectively. While the outputs of both the QSP-LCU+ROAA and the coherent one-shot algorithms agree well with the exact results for the first few time steps, the absolute error increases roughly linearly in both cases as a function of the total simulation time $t$. This linear behavior in error vs. time is expected as a direct result of the Trotterization because, in Eq. \eqref{eq:approximateTimeEvolution}, $\epsilon$ accumulates linearly with $t$ as $\epsilon \sim t^2/L \sim t \Delta t$ for fixed $\Delta t = t/L$. To determine this intrinsic Trotterization error, the ideal Trotterized evolution operators were computed using the formula in Eq.~(\ref{eq:approximateTimeEvolution}) by numerically calculating $e^{-i\mathcal{H}\Delta t}$ for each individual Trotter step exactly. A linear trendline fitted to the error of this ideal Trotterization is also shown in Fig. \ref{fig:timeDependent}, alongside similar trendlines for the total errors of the QSP-LCU+ROAA and one-shot results.

A lower slope of the error accumulation trendline reflects a higher quality approximation to the $e^{-i\mathcal{H}(k \Delta t)\Delta t}$ operators in each time step. Because the error trendline for our coherent one-shot algorithm has a lower slope than the trendline for the QSP-LCU+ROAA method, we conclude that, for this particular simulation, our coherent one-shot algorithm introduces less error at each time step than QSP-LCU+ROAA. 
%The error trendline for the coherent one-shot results is seen to follow the ideal Trotterization error trendline very closely, indicating that the coherent one-shot algorithm introduces minimal QSP error on top of the Trotterization error. In contrast, the error trendline for QSP-LCU+AA increases significantly more rapidly than that of the ideal Trotterization, due to the substantial QSP error introduced by the QSP-LCU+AA algorithm at each step. 
Relative to the ideal Trotterization error, the fitted error for one-shot at the end of the simulation period ($t = 12$) reached a value of 0.123, demonstrating roughly a factor of 1.5 improvement over the corresponding relative error 0.187 of the QSP-LCU+ROAA algorithm. Note further that by decreasing the step size $\Delta t$, the intrinsic Trotterization error can be arbitrarily minimized to achieve the desired simulation accuracy.

%We expect that at very small $\Delta t$, the QSP error due to the QSP-LCU+AA algorithm will improve, since the cosine and sine functions can be approximated well with low-degree polynomials in this regime (for small $\Delta t$ and $x\in [0,1]$, we can use the truncated Taylor expansions $\cos(x\Delta t) \approx 1-\frac{1}{2}(x\Delta t)^2$ and $\sin(x\Delta t) \approx x\Delta t$). However, the restriction that QSP-LCU+AA only be used with very small step sizes severely limits its practicability compared to the coherent one-shot algorithm, which can produce accurate results with either small \textit{or} large time steps.

The total probabilities of success were also calculated after each Trotter step using the same techniques as in the time-independent case. Both the QSP-LCU+ROAA and coherent one-shot algorithms boasted near-unity success probabilities of 99.9\% and 98.4\% on average, respectively; both of these are expected to increase with the number of Hamiltonian queries. Such a high probability of success after each step enables the sequential execution of many Trotter evolution steps without the need for classical repetition, which is a clear advantage not available to existing incoherent QSP-based simulation techniques.

\subsection{Electronic Dynamics of Molecules}
\label{sec:electronic-dynamics}

%\YL{add references here}

Electronic dynamics plays an important role in determining the chemical reactivity and response properties of molecules interacting with external field \cite{nisoli2017attosecond}. In this section, we demonstrate the ability of our algorithms to tackle these kind of applications by simulating the local charge oscillation dynamics of a H$_2$ molecule under a minimal basis. We choose H$_2$ to benchmark the performance of our algorithms because it is the simplest molecule with correlated electronic structure. In principle, these algorithms can be extended to larger molecules beyond the capability of numerically exact classical algorithms for long-time electronic dynamics, given access to a fault-tolerant quantum computer.

We give numerical details of the electronic Hamiltonian of H$_2$ molecules in Sec. \ref{sec:ham-h2}, followed by descriptions of the initial state and observables used in the simulation. In Sec. \ref{sec:h2-results}, expectation values of proper observables are presented as a function of simulation time which reveals the local charge oscillation dynamics.

%%%%%%%%%%%%%%%%%%%%%%%%%%%%%%%%%%%%%%%%%%%%%%%%%%%%%
\subsubsection{Electronic Hamiltonian of Hydrogen Molecule}
\label{sec:ham-h2}
We obtain the electronic Hamiltonian $\mathcal{H}$ of the H$_2$ molecule with internuclear distance 0.5 \AA~ using a minimal STO-3G basis. To simulate the local charge oscillation dynamics, it is convenient to transform the molecular integrals from the canonical molecular orbital representation to localized orbitals. In this case, we use Lowdin's orthogonalized atomic orbitals \cite{lowdin1950non}.

The Jordan-Wigner mapping \cite{jordan1993paulische} is then applied to map the second-quantized electronic Hamiltonian under the localized basis to qubits, resulting in a 4-qubit Hamiltonian with a dimension of 16$\times$16. A unitary dilation (block-encoding) of this Hamiltonian is then formed using one additional ancilla qubit, in similar spirit to Eq. \eqref{eq:be-2spin}, giving a total of 5-qubit unitary with $32\times 32$ dimensions. The Pauli sum representation of the H$_2$ Hamiltonian after Jordan-Wigner mapping is documented in Appendix \ref{app:ham-h2}. Given this block-encoding, the rest of the simulation algorithm follows as in Sec. \ref{Sec:OneShotHamSim}.

%%%%%%%%%%%%%%%%%%%%%%%%%%%%%%%%%%%%%%%%%%%%%%%%%%%%%
\subsubsection{Initial State and Observables}

To simulate the local charge oscillation, we start from a Slater determinant that occupies the two spin-orbitals on the hydrogen atom H$_B$, and monitor the occupation number on atom H$_A$ as the simulation time $t$ grows (see inset of Fig. \ref{fig:h2-result}). Since the two hydrogen atoms are completely equivalent, the system behaves much like a double-well potential, and the localized charge initially on H$_B$ will migrate gradually to H$_A$ and then back to H$_B$ through tunneling. This process persists as no vibronic coupling is considered in our Hamiltonian. Note that in practice, the electronic energy may gradually dissipate to vibrations of the H$_2$ molecule. We shall leave the case of vibronic simulation to future work.

In more concrete terms, assume the 4 system qubits are in block-spin format where the first two are spin-up orbitals on hydrogen H$_A$ and H$_B$, and the last two are spin-down orbitals on H$_A$ and H$_B$. As we have chosen the initial state to have one electron in each of the spin-up and spin-down orbitals on atom H$_B$, the initial state $\ket{\psi}_{\mathrm{init}}$ is
\begin{align}
    \ket{\psi}_{\mathrm{init}} = \ket{0101},
\end{align}
where the $\ket{\cdot}$ notation on the right-hand-side represents a Slater determinant, and the qubit state $\ket{1}$ ($\ket{0}$) stands for one (no) electron occupying the corresponding spin-orbital.

To quantify the charge dynamics, we monitor the expectation value of the total occupation number $\langle \hat{n}_A \rangle$ on hydrogen H$_A$, given by
\begin{align}
    \langle\hat{n}_A\rangle = \bra{\psi}\left(a_{A, \uparrow}^\dagger a_{A, \uparrow} + a_{A, \downarrow}^\dagger a_{A, \downarrow}\right)\ket{\psi},
\end{align}
 where  $\ket{\psi} \in \{\left[e^{-i\mathcal{H}t}\right]_{\textrm{ROAA}}\ket{\psi_0},\left[e^{-i\mathcal{H}t}\right]_{\textrm{OS}}\ket{\psi_0}\}$ for increasing simulation time $t$. On quantum computers, the $\hat{n}_A$ operator is mapped to the following spin operators via the Jordan-Wigner mapping:
\begin{align}
    \hat{n}_A = \frac{I_{A, \uparrow}-Z_{A, \uparrow}}{2} + \frac{I_{A, \downarrow}-Z_{A, \downarrow}}{2},
\end{align}
where $I$ and $Z$ are the Pauli operators acting on the corresponding qubits.

%%%%%%%%%%%%%%%%%%%%%%%%%%%%%%%%%%%%%%%%%%%%%%%%%%%%%
\subsubsection{Results}
\label{sec:h2-results}

Fig. \ref{fig:h2-result} shows the time-dependence of $\langle \hat{n}_A \rangle$ after the initial state $\ket{\psi}_{\mathrm{init}}$ is time-evolved using both the QSP-LCU+ROAA and coherent one-shot algorithms. To benchmark these numerical results, the values of $\langle \hat{n}_A \rangle$ corresponding to exact time-evolution (computed via Eq. (\ref{eq:exact_evolution})) are included in the plot as well. As in the earlier Heisenberg model simulations, the degrees of the QSP polynomials used in the two schemes were chosen so that the number of queries made to the Hamiltonian block encoding by each algorithm would be nearly equivalent. Explicitly, polynomials with degrees $d_{\textrm{cos}} = 6$ and $d_{\textrm{sin}} = 5$ were used to approximate $\cos(\mathcal{H}t)$ and  $-i\sin(\mathcal{H}t)$, respectively, in the QSP-LCU+ROAA scheme; a polynomial with degree $d_{\textrm{EECE}} = 32$ was used in the coherent one-shot scheme. Hence, QSP-LCU+ROAA employs 33 queries to the block-encoding per simulation interval, while one-shot employs 32. The QSP phases corresponding to these polynomials were obtained (as in the Heisenberg model simulations) using the PyQSP optimization package with the parameters $\alpha = 1$ and $\beta = 0.5$.

As anticipated, the total occupation number on atom H$_A$ gradually increases from 0 to near 2, before charge is transferred back to atom H$_B$. This oscillation has a period of roughly 0.12 fs, which lies in the typical regime of attosecond electronic dynamics. For this particular simulation, our coherent one-shot algorithm produced results with lower error than those of QSP-LCU+ROAA: over the simulation intervals tested, one-shot resulted in an average population error of 0.064, compared to 0.094 for QSP-LCU+ROAA.

Note that for times $t < 0.06\textrm{ fs}$, both algorithms were found to succeed with reasonably high probability (namely, with success probabilities $1-\delta > 0.85$). However, for some simulation intervals after $t=0.06\textrm{ fs}$, the success probabilities of both algorithms dropped below this 0.85 threshold. These data points corresponding to lower success probability are distinguished with dotted markers in Fig. \ref{fig:h2-result}. Decreasing success probability is expected at long simulation times, since the QSP polynomial approximations to the functions $\cos(x\tau)$ and $\sin(x\tau)$ (in the case of QSP-LCU+ROAA) and $e^{-ix\tau}$ (in the case of coherent one-shot) deteriorate for large $\tau$ if the QSP polynomial degrees are held constant. If a higher success probability is desired at long simulation times, the success probability can be arbitrarily improved to satisfaction by increasing the degrees of the QSP polynomials used.

\begin{figure*}[htbp]
  \includegraphics[width=14cm,height=7cm]{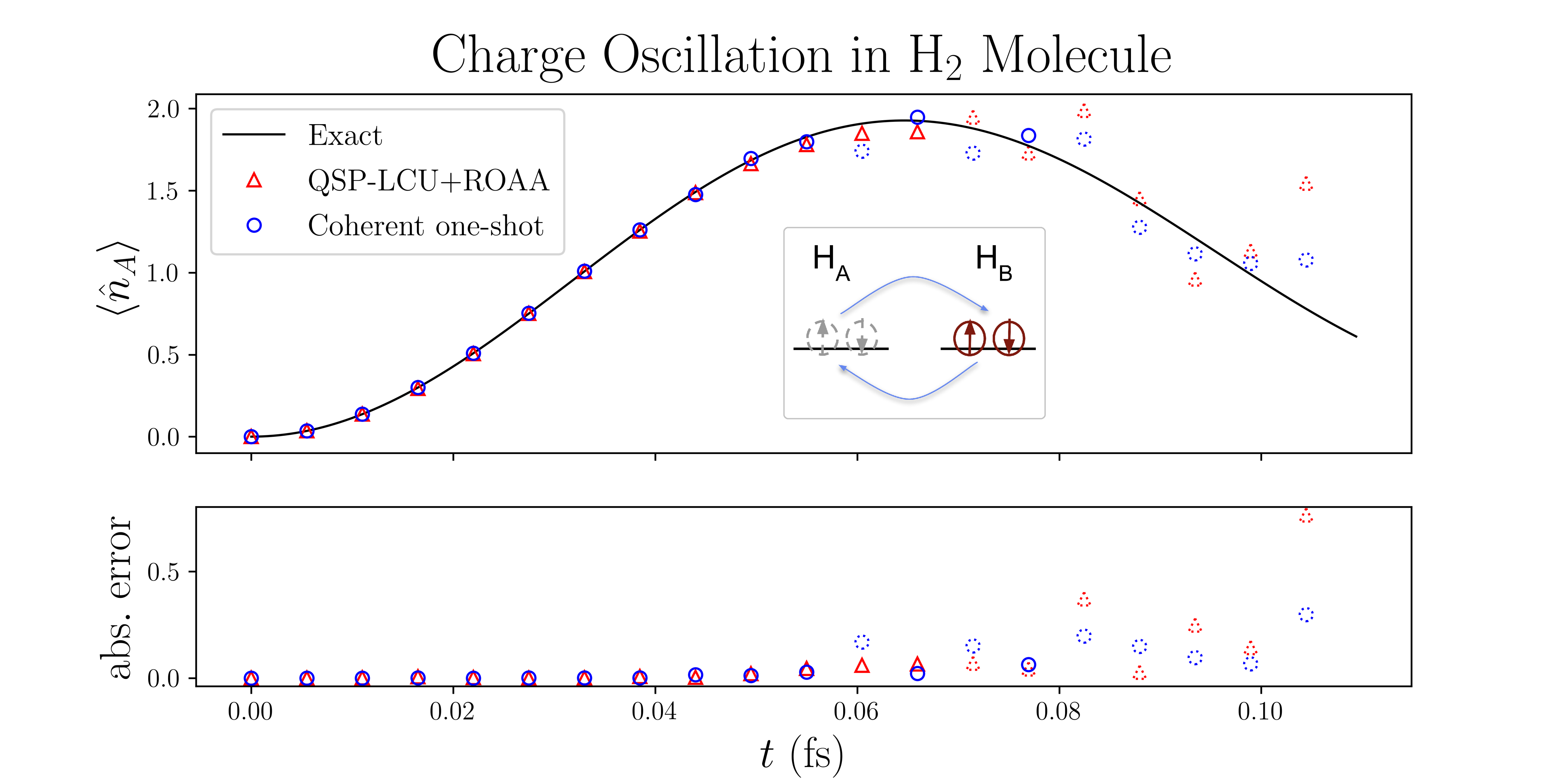}
  \caption{Local charge oscillation dynamics in an H$_2$ molecule with an internuclear distance of 0.5~\AA\ simulated using both the conventional LCU algorithm with robust oblivious amplitude amplification (QSP-LCU+ROAA) and the coherent one-shot algorithm. The time scale on the horizontal axis is given in femtoseconds. The QSP-LCU+ROAA scheme made 33 queries to the Hamiltonian block-encoding for each simulation interval, while the one-shot scheme made 32. The inset shows the arrangement of the localized orbitals and initial states. \textbf{(Top)} Expectation value of the total electron occupation number $\langle \hat{n}_A\rangle$ on atom H$_{A}$. \textbf{(Bottom)} Absolute error between algorithm output and exact result at each simulation time interval. Note for both plots that data points associated with success probability $1-\delta > 0.85$ are presented with solid markers, while points with $1-\delta \leq 0.85$ are instead presented with dotted markers.}
    \label{fig:h2-result}
\end{figure*}

\section{Discussion and Conclusion}\label{Sec:Discussion}
In this work, we have presented the concept of an \emph{efficient fully-coherent} Hamiltonian simulation algorithm, which performs Hamiltonian simulation and retains the time evolved wave function with arbitrarily high success probability, while also achieving a query complexity polynomial in the effective time $\| \mathcal{H} \| |t|$, the logarithm of the inverse error $\ln(1/\epsilon)$, and the logarithm of the inverse failure probability $\ln(1/\delta)$. We focused on developing said algorithms through QSP techniques, provided a block encoding of the Hamiltonian. We showed in Sec.~\ref{Sec:ConventionalHamSim} that conventional QSP-based simulation is not fully-coherent. However, it can be made fully-coherent by augmentation with amplitude amplification as in Sec.~\ref{sec:conv_ampamp} at the expense of appending a $\ln(1/\delta)$ multiplicative factor to the query complexity, or by augmentation with robust oblivious amplitude amplification as in Sec.~\ref{sec:RO_ampamp}, which circumvents this factor.

We further introduced the concept of a pre-transformation in Sec.~\ref{Sec:OneShotHamSim} and used it to develop our coherent one-shot Hamiltonian simulation algorithm, which achieves efficient fully-coherent simulation and attains a query complexity additive in $\ln(1/\delta)$. In Sec.~\ref{Sec:Scaling}, we compared the query complexities of the coherent algorithms, illustrating a trade-off between the coherent one-shot algorithm and QSP-based simulation augmented with robust oblivious amplitude amplification. Finally, the applications of our algorithms in Sec.~\ref{Sec:Applications} to the Heisenberg model and an H$_2$ molecule under a minimal basis evidence these fully coherent algorithms as powerful tools for the simulation of spin chains and correlated electronic dynamics.

As established in Sec. \ref{Sec:Applications}, using a fixed number of queries to the Hamiltonian, the coherent one-shot algorithm is shown to efficiently simulate time evolution for both small and large values of times; thus, it is a highly effective tool for both time-independent Hamiltonian simulations with a single long-time algorithm call \emph{and} Trotterized time-dependent simulations with many short-time algorithm calls. 

At the heart of coherent one-shot Hamiltonian simulation is the use of a pre-transformation to rescale the eigenvalues of the Hamiltonian. This is crucial as the complex exponential cannot be implemented as a QSP polynomial for arbitrary inputs, but can be done over a restricted range as we analytically verified by the polynomial constructed in Sec.~\ref{Sec:OneShotHamSim}. Looking towards extensions of this idea, we expect pre-transformations to generate a broader class of achievable transformations, and also note that out framework of compressing the spectrum and then applying an even/odd extension of a function could certainly be generalized to functions other than the complex exponential. In addition, while we employed a linear pre-transformation here, it would be interesting to study the applicability of nonlinear pre-transformations. One way to construct such a pre-transformation would be via QSP itself, which suggests the power of concatenated QSP algorithms -- algorithms that apply a sequence of QSP transformations.

Regarding improvements of the coherent one-shot algorithm, it would be highly desirable to attain a better analytic expression for a QSP polynomial that approximates $e^{-ix \tau}$ over a range of positive $x$. While the construction provided in Sec.~\ref{Sec:OneShotHamSim} is sufficient, it is by no means optimal. Additionally, the query complexity $N_{\mathcal{H}}^{\mathrm{OS}}$ could likely be sharpened to better elucidate the tradeoff between the coherent one-shot algorithm and the QSP-LCU+ROAA algorithm.

Moreover, one interesting application of the fully-coherent simulation algorithms discussed here is to time-dependent Hamiltonians, as discussed in Sec.~\ref{sec:time-dependent}. However, the lower-bound on errors in this case is limited by the Trotterization error, which suggests the utility of studying the coherent simulation of time-dependent Hamiltonians beyond Trotterization. For instance, it may be fruitful to combine a fully-coherent simulation algorithm with the time-dependent methods presented in Refs.~\cite{low2018hamiltonian,hen2021quantum}. In addition, a fully-coherent algorithm could also be integrated with higher-order product formulas for the time evolution operator~\cite{childs2021theory}, or even invoked as a subroutine in specialized algorithms for local Hamiltonian simulation~\cite{Haah_2021}.

We also note that efficient block-encoding of correlated electronic Hamiltonians of large systems is paramount to the success of QSP-based algorithms. We look forward to novel constructions of more efficient block-encodings for correlated electronic structure Hamiltonians beyond the LCU framework \cite{Childs_2017,wan2021block} and tensor-hyper contraction \cite{lee2021even}.

We hope that this work will shed light on the importance of full-coherence in quantum algorithms. The fully-coherent Hamiltonian simulation algorithms presented here should be applicable as subroutines in larger quantum algorithms, where reliable passing of \textit{quantum} data from one subroutine to another is required without measurement or post-selection. In this sense, the coherence of our Hamiltonian simulation algorithm, and hopefully many more to come, can facilitate integration with other quantum primitives and give rise to further novel and powerful quantum algorithms. Moreover, in quantum dynamical processes of many physical systems, we expect the system itself passes quantum information coherently in some form from one part to the other. Because of this, we believe such capability to pass quantum data coherently from one subroutine to another in quantum simulation algorithms is crucial to model the quantum dynamics in a variety of chemical, biology, and condensed matter systems with ever increasing complexity. It is also our wish that this paper emphasizes the power and flexibility of QSP techniques for quantum simulation and beyond. As QSP has already revolutionized the design of quantum algorithms for problems of physical relevance, it is likely to also provide novel pathways to many applications of chemical relevance.

\begin{acknowledgments}
The authors thank Minh Tran and Brenda Rubenstein for helpful discussions. YL and ZEC were supported in part by NTT Research. The work of YL and ILC on analysis and numerical simulation was supported by the U.S. Department of Energy, Office of Science, National Quantum Information Science Research Centers, Co-Design Center for Quantum Advantage, under contract number DE-SC0012704.

\textit{Conflicts of Interest:} The authors have no conflicts to disclose.

\textit{Data Availability:} The data in this paper were generated using the PyQSP package~\cite{pyqsp}. Data supporting then findings of this study are available from the corresponding author upon reasonable request.

\end{acknowledgments}

\appendix

\section{A Foray into QSP, QET, and QSVT}\label{Sec:qsp}
In this appendix, we review the technique of quantum signal processing, presenting the theorems spelled out in Refs. \cite{Low_2016, Low_2017, Low_2019}. From a high level, QSP provides a systematic method to apply a nearly arbitrary polynomial transformation to a quantum subsystem. For a target polynomial of degree $d$, this is achieved using $\mathcal{O}(d)$ elementary quantum gates. In particular, one applies a specific sequence of $(d+1)$ SU(2) rotations to the subsystem of interest, where each rotation is parameterized by an angle $\phi_k \in \mathbb{R}$. The QSP algorithm is parametric in that the polynomial transformation achieved is completely characterized by the choice of phase angles $\{\phi_k\}$.

Subsequently, we will discuss how QSP may be naturally generalized to apply a polynomial transformation to the eigenvalues of a matrix encoded in a block of a unitary matrix, in which case this procedure is known as a quantum eigenvalue transform (QET) and ultimately produces a polynomial transform of a matrix \cite{Low_2017, Low_2019, martyn2021grand}. Similarly, this procedure may be further extended to the singular values of a possibly non-square matrix, in which case this is known as a quantum singular value transformation (QSVT), and the polynomial is applied to the singular values of the matrix \cite{Gily_n_2019}. This ability to polynomially transform the singular values of an arbitrary linear operator is incredibly powerful and has been shown to be a building block from which of nearly all known quantum algorithms may be derived \cite{Gily_n_2019, martyn2021grand}.

\subsection{Overview of QSP}\label{Sec:qsp_Overview}
Quantum signal processing (QSP) \cite{Low_2016, Low_2017, Low_2019} works by interleaving a signal rotation operator $W$, and a signal processing rotation operator $S$. These operators are taken to be SU(2) rotations about different axes, where the signal rotation operator is a rotation through a fixed angle $\theta$ and the signal processing rotation operator is a rotation through a variable angle parametrized by a real number $\phi$. Typically, $W$ is taken to be an $x$-rotation and $S$ a $z$-rotation, although other equivalent conventions exist. 
 
Let us follow standard conventions and define the signal rotation operator as
\begin{equation}
    W(x) = \begin{bmatrix}
        x & i\sqrt{1-x^2} \\
        i\sqrt{1-x^2} & x
    \end{bmatrix},
\end{equation}
which is an $x$-rotation through an angle $\theta = -2\cos^{-1}x$ ($x \in [-1,1]$). Similarly, define the signal processing rotation operator as
\begin{equation}
    S(\phi) = e^{i\phi Z},
\end{equation}
which is a $z$-rotation through an angle $-2\phi$. Then, with a set of \emph{QSP phases} $\vec{\phi} = (\phi_0, \phi_1, ... , \phi_d) \in \mathbb{R}^{d+1}$, one can construct the QSP sequence, $U^{\vec{\phi}}$, which is defined as the following interspersed sequence of $W$ and $S$:
\begin{equation}
    \begin{gathered}
        U^{\vec{\phi}} = S(\phi_0) \prod_{k=1}^d  W(x)  S(\phi_k) = e^{i \phi_0 Z} \prod_{k=1}^d  W(x)  e^{i\phi_k Z}.
    \end{gathered}
\end{equation}

Curiously, the matrix elements of this operation are polynomials of $x$, parameterized by the QSP phases $\vec{\phi}$. In particular, Ref. \cite{Low_2016} proved that 
\begin{equation}\label{eq:qsp}
    \begin{gathered}
        U^{\vec{\phi}} = e^{i \phi_0 Z} \prod_{k=1}^d  W(a)  e^{i\phi_k Z} = \\
        \begin{bmatrix}
            P(x) & iQ(x)\sqrt{1-x^2} \\
            iQ^*(x)\sqrt{1-x^2} & P^*(x)
        \end{bmatrix}
    \end{gathered}
\end{equation}
where $P(x)$ and $Q(x)$ are polynomials that obey:
\begin{equation}\label{eq:qsp_conditions}
    \begin{split}
        & \text{1. } {\rm deg}(P) \leq d, \ {\rm deg}(Q) \leq d-1 \\
        & \text{2. } P(x)\ \text{has parity } d~ {\rm mod}~ 2 \\
        & \text{3. } |P(x)|^2 + (1-x^2) |Q(x)|^2 = 1, \ \forall \  x \in [-1,1]. 
    \end{split}
\end{equation}

This result is quite fascinating. It tells us that we can prepare polynomials of $x$ by projecting into a block of $U^{\vec{\phi}}$. For instance, $\langle 0| U^{\vec{\phi}} | 0\rangle = P(x)$, so $P(x)$ dictates the probability a particle in state $|0\rangle$ remains in its initial state upon application of $U^{\vec{\phi}}$. 

The block encoding has thus far focused on the $|0\rangle \langle 0 |$ matrix elements: $x = \langle 0| W(x)|0\rangle$, and $P(x) = \langle 0| U^{\vec{\phi}}|0\rangle$. While this choice seems natural, it is not necessary and is actually inhibitory. In particular, matrix elements of $U^{\vec{\phi}}$ in bases other than the computational basis may be projected out, in which case these matrix elements are linear combinations of $P(x)$ and $Q(x)\sqrt{1-x^2}$ (and their complex conjugates). For instance, we may project into the $\{|+\rangle, |-\rangle \}$ basis, wherein we find the useful result 
\begin{equation}
\begin{split}
    &\langle+|U^{\vec{\phi}}|+\rangle = \text{Re}(P(x)) + i\text{Re}(Q(x))\sqrt{1-x^2}.
\end{split}
\end{equation}
This result is important because projecting into the $|+\rangle\langle+|$ matrix element allows us to construct a larger class of polynomials than those achieved by projection into the computational basis. Whereas $P(x)$ is required to obey $|P(x=\pm 1)| = 1$ as per the third condition in Eq.~(\ref{eq:qsp_conditions}), the above polynomial needs not obey this restriction. In particular, it can be shown that the projection into the $|+\rangle \langle +|$ matrix element can accurately approximate any real polynomial with parity $d~ {\rm mod}~ 2$ such that $\deg({\rm Poly}) \le d$, and $| {\rm Poly}(a) | \le 1 ~ \forall ~x \in [-1, 1]$. This can be achieved by selecting an appropriate $P$ whose real part approximates the desired function and a $Q$ with small real component. As we will see later, this additional polynomial freedom will be crucial to conventional QSP-based Hamiltonian simulation. 

Implicit in this construction is that we can access the correct block of the QSP sequence, i.e. we can choose to access some matrix element such as $\langle 0|U^{\vec{\phi}}|0\rangle$. Accessing the correct block is fundamental to applications of QSP, as QSP-based algorithms will only be successful if a desired block is accessed. In such an application, we must project the final state into being in the desired block by performing a projective measurement. For instance, if our goal is to access the $|0\rangle \langle 0 |$ block of the QSP sequence, we may measure the auxiliary qubit (used to achieve the block encoding) in the computational basis, which will isolate the desired polynomial if $|0\rangle$ is measured.

In general however, this measurement will output the desired result with non-unit probability $p=|\langle 0 |P(x) | 0\rangle|^2$, which is dictated by the choice of polynomial. Nonetheless, the probability of accessing the correct block can be increased by using classical repetition or amplitude amplification. Although this procedure is implicitly performed in QSP-based algorithms, it has the downside that it requires multiple instance of QSP. As a result, if one desires a probability of success $\geq 1-\delta$, then there is a $\delta$-dependent multiplicative factor in the complexity of the QSP-based algorithm, and this factor is often ignored in statements of query complexity. For example, in the case of amplitude amplification, this factor is $\Theta\left(\ln(\frac{1}{\delta})\right)$. In Sec.~\ref{Sec:OneShotHamSim}, we eliminated this multiplicative factor for QSP-based Hamiltonian simulation by presenting a fully-coherent simulation algorithm that succeeds with arbitrarily high probability and requires only a single QSP instantiation, while retaining a near-optimal scaling of the query complexity.

\subsection{Generalization to QET} \label{sec:qsp_Gen}
The construction of QSP may be generalized to apply polynomial transformations to the eigenvalues of a matrix, a procedure that we call a \emph{quantum eigenvalue transformation} (QET). To see this, note that QSP has the following interpretation: QSP begins with a matrix $W(x)$ that is effectively a block encoding of $x$,
\begin{equation}
    W(x) = \begin{bmatrix}
        x & \cdot \\
        \cdot & \cdot
    \end{bmatrix}.
\end{equation}
By interspersing this matrix with rotation operators parameterized by $\vec{\phi}$, we constructed an operator that block encodes a polynomial transformation of $x$, $P(x)$:
\begin{equation}
    U^{\vec{\phi}} = \begin{bmatrix}
        P(x) & \cdot \\
        \cdot & \cdot
    \end{bmatrix}.
\end{equation}
So in terms of block encodings, this procedure performs the operation $x \mapsto P(x)$. 

Paralleling this scenario, one may instead begin with a unitary block encoding $U$ of a matrix $A = \sum_\lambda \lambda |\lambda \rangle \langle \lambda|$:
\begin{equation}
    U=\begin{bmatrix}
        A & \cdot \\
        \cdot & \cdot
    \end{bmatrix}
\end{equation}
(this is only possible if $\|A\| \leq 1$, which we assume here; if this condition is not met, one can rescale $A$ by a constant to meet this condition and block encode this rescaled matrix.) More generally, we assume that $A$ can be accessed with a projector $\Pi$ as $A=\Pi U \Pi$. 

In QET, one intersperses $U$ with a rotation operator that acts as a $z$-rotation within each eigenspace of $A$. This operator, known as the projector-controlled phase shift, may be expressed as $\Pi_\phi := e^{i\phi(2\Pi-I)}$ for a rotation angle $\phi$, and can be straightforwardly constructed as per Fig.~\ref{fig:ProjControlledPhaseShift}~\cite{Gily_n_2019}.
\begin{figure}[htbp]
    \begin{center}
    \includegraphics[width=4.9cm]{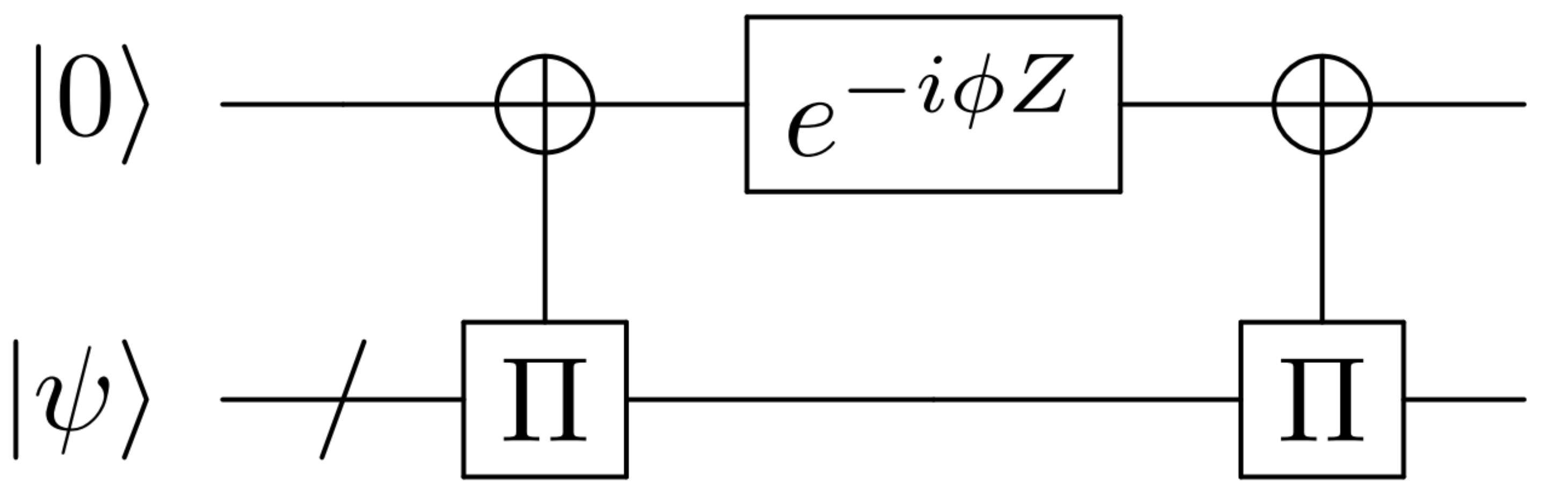}
    \end{center}
\caption{The circuit used to realize the projector-controlled phase shift $\Pi_\phi$ (up to a global phase), as a product of projector-controlled-{\sc not} gates and a $z$-rotation.}
\label{fig:ProjControlledPhaseShift}
\end{figure}

Paralleling Eq.~(\ref{eq:qsp}), one can then define the following QET sequence, which effectively performs QSP within each eigenspace and a polynomial transformation of $A$: 
\begin{equation}
\begin{split}
    U^{\vec{\phi}} = \Pi_{\phi_{0}}& \prod_{k=1}^{d} U \Pi_{\phi_{k}} \ \  \Rightarrow \\
    &\Pi \  U^{\vec{\phi}} \Pi = \sum_\lambda P(\lambda) |\lambda \rangle \langle \lambda| = P(A).
\end{split}
\end{equation}

In terms of block encodings, QET maps a matrix $A \mapsto P(A)$, thus developing a polynomial transformation of the matrix. As in Sec. \ref{Sec:qsp_Overview}, it is conventional that $\Pi = |0\rangle \langle 0|$, in which case we may obtain $P(A)$ when we project into the $|0\rangle \langle 0|$ matrix element of $U^{\vec{\phi}}$, which can be done with some finite probability by measuring the ancilla qubit used to achieve the encoding. Alternatively. we may also project into the $|+\rangle \langle +|$ component to achieve a larger class of polynomial transformations of $A$.

\subsection{Generalization to QSVT}\label{sec:QSVT_overview}
Generalizing even further, this construction may be adapted to apply a polynomial transformation to the singular values of an arbitrary linear operations (i.e. a possibly non-square matrix). Recall that we may write the singular value decomposition of a matrix $A$ as $A = W\Sigma V^\dag$, where $W$ and $V$ are unitaries and $\Sigma$ contains the singular values (non-negative) along its diagonal. Denoting the singular values by $\{\sigma_i\}$, this may be re-expressed as
\begin{equation}
    A = \sum_i \sigma_i |w_i\rangle \langle v_i|, 
\end{equation}
where $|w_i \rangle$ and $| v_i \rangle$ are the columns of $W$ and $V$ respectively. 

Suppose we have a unitary $U$ that block encodes $A$ as $\tilde{\Pi} U \Pi = A$, where $\tilde{\Pi}$ and $\Pi$ are projectors. We may follow a procedure analogous to that of QET by interspersing this block encoding and its Hermitian conjugate (a block encoding of $A^\dag)$ with an operator that acts a rotation within the left and right singular vector spaces (where the left singular vector space is spanned by $\{|w_k\rangle\}$, and the right singular vector space is spanned by $\{|v_k\rangle\}$). As before, these are projector-controlled phase shifts, denoted by $\Pi_\phi$ and $\tilde{\Pi}_\phi$, respectively, and can be constructed analogous to Fig.~\ref{fig:ProjControlledPhaseShift}.

We then define the following QSVT sequence, which performs QSP within each singular vector space, and thus applies a polynomial transformation to the singular values. Its form differs from that of QSP and QET, as it must deal with the different left and right singular vector spaces: 
\begin{equation}
\begin{split}
    &U^{\vec{\phi}} = \begin{cases}
        \tilde{\Pi}_{\phi_{1}}  U \left[ \prod_{k=1}^{(d-1)/2} \Pi_{\phi_{2k}} U^\dagger \tilde{\Pi}_{\phi_{2k+1}}  U  \right]  & d \ \text{odd} \\
        \left[ \prod_{k=1}^{d/2} \Pi_{\phi_{2k-1}} U^\dagger \tilde{\Pi}_{\phi_{2k}}  U  \right] & d \ \text{even},
    \end{cases} \\
    &\Rightarrow \quad P^{(SV)}(A) = \begin{cases}
        \tilde{\Pi} U^{\vec{\phi}} \Pi & d \ \text{odd}  \\ 
        \Pi U^{\vec{\phi}} \Pi & d \ \text{even},
    \end{cases}
\end{split}
\end{equation}
where
\begin{equation}
    P^{(SV)}(A) := \begin{cases}
        \sum_i P(\sigma_i) |w_i\rangle \langle v_i| & P \text{ odd,} \\
        \sum_i P(\sigma_i) |v_i\rangle \langle v_i| & P \text{ even.} 
    \end{cases}
\end{equation}
This protocol, known as the \emph{quantum singular value transformation} (QSVT), is a powerful subroutine that has been demonstrated to underlie nearly all quantum algorithms \cite{Gily_n_2019, martyn2021grand}. Note that $ P^{(SV)}(A)$ differs for odd and even polynomials through the basis in which it is expressed. This arises because each application of $A$ and $A^\dag$ switches between the left and right singular vector spaces and vice versa, respectively, and the number of applications directly corresponds to the parity of the polynomial. This crucial distinction comes into play in Sec.~\ref{sec:full-coherent-aa}.

Ultimately, QSP is at the crux of QET and QSVT. So if we can develop a useful polynomial through QSP, say one that accurately approximates a function of interest, then we may apply it to the eigenvalues of a matrix via QET, or to the singular values via QSVT. A particularly useful such function is the complex exponential for Hamiltonian simulation, whose construction we discuss in Secs.~\ref{Sec:ConventionalHamSim} and~\ref{Sec:OneShotHamSim}.

\section{Polynomial Approximation of the Sign Function} \label{Sec:SignFunction}

Here we outline the construction of a polynomial approximation to the sign function. We closely follow the development of Refs. \cite{low2017quantum, Sachdeva_Faster}, wherein the full details and proofs of some statements employed here can be found. 

To start the construction, we note that Ref. \cite{low2017quantum} proves that $\text{erf}(kx) = \frac{2}{\sqrt{\pi}} \int_0^{kx} e^{-y^2} dy = \frac{2k}{\sqrt{\pi}} \int_0^x e^{-(ku)^2} du$ is an $\epsilon$-approximation to the sign function for $|x|\geq \Delta/2$ if $k= \frac{\sqrt{2}}{\Delta} \sqrt{W(\frac{2}{\pi \epsilon^2})} < \frac{\sqrt{2}}{\Delta} \sqrt{\ln(\frac{2}{\pi \epsilon^2})}$ (for $\frac{2}{\pi \epsilon^2} > e$), where $W(x)$ is the Lambert $W$ function. A polynomial approximation to the sign function may then be built by first constructing a polynomial approximation to $e^{-k^2 u^2}$, inserting it into this expression, and integrating as necessary. 

The exponential function $e^{-a(x+1)}$ (which is bounded in magnitude by $1$ for $x \in [-1,1]$) may be approximated for $x \in [-1,1]$ by truncating the following modified Jacobi-Anger expansion:
\begin{equation}
    e^{-a(x+1)} = e^{-a}\Big(I_0(a) + 2\sum_{j=1}^\infty I_j(a) T_j(-x) \Big),
\end{equation}
where $I_j(x)$ is the modified Bessel function (of the first kind) of order $j$, and $T_j(x)$ is the Chebyshev polynomial of order $j$. Ref.~\cite{low2017quantum} proves that one can attain a polynomial $\epsilon$-approximation to $e^{-a(x+1)}$ by truncating this series at $j \geq  \sqrt{2\cdot  \text{max} (ae^2, \ln(2/\epsilon)) \cdot \ln(4/\epsilon)}$, which is thus a polynomial of degree at least $d_{\text{exp},a,\epsilon}:= \sqrt{2 \cdot \text{max} (ae^2, \ln(2/\epsilon)) \cdot \ln(4/\epsilon)}$. It will be convenient to express this as $\epsilon_{\text{exp},a,d} \leq \epsilon $ for $d \geq d_{\text{exp},a,\epsilon}$ -- i.e., the error suffered is less than $\epsilon$ when the degree is greater than $ d_{\text{exp},a,\epsilon}$. 

Ref.~\cite{low2017quantum} further demonstrates how this construction may be used to estimate $e^{-(ku)^2}$ with a simple change of variables. Inserting the resulting polynomial into the expression $\text{erf}(kx) = \frac{2k}{\sqrt{\pi}} \int_0^x e^{-(ku)^2} du$, we obtain a polynomial approximation to the $\text{erf}(kx)$ that suffers error $\epsilon_{\text{erf},k,d} \leq \frac{4k}{\sqrt{\pi}d} \epsilon_{\text{exp},k^2/2,(d-1)/2}$. 

\begin{figure*}[htbp]
  \includegraphics[width=0.95\linewidth]{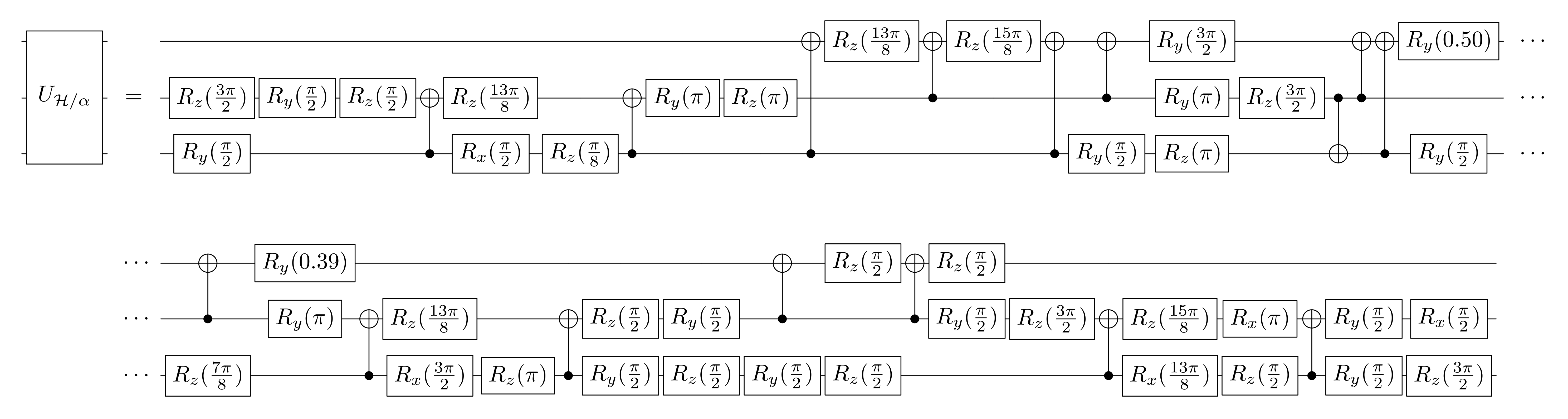}
  \caption{Quantum circuit implementing the Hamiltonian block encoding $U_{\mathcal{H}/\alpha}$ of the two-qubit Heisenberg model simulated in Sec.~\ref{sec:numsim-t-ind}. From top to bottom, the qubits in the circuit are ordered from most to least significant.}
  \label{fig:unitaryBlockEncodingCirc}
\end{figure*}

Ultimately, we desire an $\epsilon/2$-approximation to $\text{erf}(kx)$ with $k=\frac{\sqrt{2}}{\Delta} \sqrt{W\left(\frac{2}{\pi (\epsilon/2)^2} \right)} = \frac{\sqrt{2}}{\Delta} \sqrt{W\left(\frac{8}{\pi \epsilon^2} \right)}$, such that this function provides an $\epsilon$-approximation to the sign function by the triangle inequality. Thus, it will suffice to choose 
\begin{equation}
    \begin{split}
        \frac{4k}{\sqrt{\pi}d} &\epsilon_{\text{exp},k^2/2,(d-1)/2} \leq \epsilon/2 \Rightarrow \\
        &\qquad \epsilon_{\text{exp},k^2/2,(d-1)/2} \leq \frac{\sqrt{\pi}}{8}\frac{\epsilon d}{k},
    \end{split}    
\end{equation}
or equivalently, 
\begin{equation}
    \begin{aligned}
        &(d-1)/2 \geq d_{\text{exp},k^2/2,\frac{\sqrt{\pi}}{8} \frac{\epsilon d}{k}}= \\
        & \sqrt{2 \cdot \text{max}\Bigg(\frac{e^2k^2}{2}, \ \ln(\frac{16 k}{\sqrt{\pi}\epsilon d}) \Bigg) \cdot \ln(\frac{32 k}{\sqrt{\pi}\epsilon d})}. \qquad \qquad
    \end{aligned}
    \label{eq:d-1bound}
\end{equation}

If the first term is larger, then this bound requires
\begin{equation}
    \begin{split}
        &(d-1)/2 \geq  ek \sqrt{\ln(\frac{32 k}{\sqrt{\pi}\epsilon}) + \ln(\frac{1}{d})}.
    \end{split}
\end{equation}
Using the variable substitution $v=(d-1)/2$, squaring this equation and using $\frac{1}{1+2v} < \frac{1}{2v}$, we find that it will suffice to choose a $v^2 \geq \frac{1}{2} e^2k^2 W(\frac{512}{\pi e^2} \frac{1}{\epsilon^2})$, where $W(x)$ is the Lambert $W$ function. Equivalently,
\begin{equation}
    \begin{split}
        d \geq \frac{2 e}{\Delta} \sqrt{W\left(\frac{8}{\pi \epsilon^2} \right) W\left(\frac{512}{e^2 \pi} \frac{1}{\epsilon^2} \right)} + 1,
    \end{split}
\end{equation}
where we have inserted the explicit expression for $k$.

On the other hand, if the second term in Eq.~(\ref{eq:d-1bound}) is larger, then the bound requires
\begin{equation}
    \begin{split}
        &(d-1)/2 \geq \sqrt{2 \cdot \ln(\frac{16k}{\sqrt{\pi}\epsilon d})\left( \ln 2 + \ln(\frac{16k}{\sqrt{\pi}\epsilon d}) \right)}.
    \end{split} 
\end{equation}
Again setting $v=(d-1)/2$, and using the bound $\sqrt{x^2 +ax} \leq x + a/2$, we find that this can be satisfied if 
\begin{equation}
        d \geq 2\sqrt{2} W\left( \frac{8\sqrt{2}}{\sqrt{\pi} \Delta \epsilon} \sqrt{W\left(\frac{8}{\pi \epsilon^2} \right)} \right) + 1.
\end{equation}

In general then, it will suffice to truncate at degree 
\begin{widetext}
    \begin{equation}\label{eq:gamma}
        d = 2 \cdot \Bigg\lceil \text{max}\left(\frac{e}{\Delta} \sqrt{W\left(\frac{8}{\pi \epsilon^2} \right) W\left(\frac{512}{e^2 \pi} \frac{1}{\epsilon^2} \right)}, \  \sqrt{2} W\left( \frac{8\sqrt{2}}{\sqrt{\pi} \Delta \epsilon} \sqrt{W\left(\frac{8}{\pi \epsilon^2} \right)} \right) \right) \Bigg\rceil + 1 = : \gamma(\epsilon, \Delta),
    \end{equation}
\end{widetext}
which is necessarily an odd integer.

So in summary, an $\epsilon$ approximation to the sign function, valid for $|x|\geq \Delta/2$, can be constructed as polynomial of odd degree $\gamma(\epsilon, \Delta)$. Noting that $W(x) = \ln x - \ln (\ln x) + o(1) = \Theta\left(\ln(x)\right)$ for large $x$, we see that $\gamma(\epsilon,\Delta)$ scales as $\gamma(\epsilon,\Delta) = \Theta \left( \frac{1}{\Delta} \ln(\frac{1}{\epsilon})\right)$, corroborating the claims in the literature~\cite{low2017quantum}.

\section{Quantum circuit for unitary block-encoding}\label{sec:unitaryBlockEncodingCirc}

For completeness, the gate-level Hamiltonian block-encoding circuit for the two-spin Heisenberg model simulated in Sec.~\ref{sec:numsim-t-ind} is provided in Fig. \ref{fig:unitaryBlockEncodingCirc}. This circuit implements the unitary $U_{\mathcal{H}/\alpha}$ defined in Eq.~(\ref{eq:be-2spin}) with $\alpha = 1.5$. Recall that this specific block-encoded Hamiltonian $\mathcal{H}$ is defined according to Eqs.~(\ref{eq:heisenberg-Hfull}-\ref{eq:heisenberg-H1}) with the parameters $h_1(t) = h_2(t) = 0.5$ for all times $t$ and $g_x = 1$, $g_y = g_z = 0$. This quantum circuit for $U_{\mathcal{H}/\alpha}$ was generated using the Quantum Shannon Decomposition method via the UniversalQCompiler software package in Mathematica \cite{universalqcompiler}.

\section{Pauli Sum Representation of Electronic Hamiltonian of H$_2$}
\label{app:ham-h2}
The Pauli sum representation of the H$_2$ electronic Hamiltonian are provided below in Table~\ref{tab:h2_paulisum}.

\begin{table}[ht!]
    \centering
    \begin{tabular}{c|c}
        \hline\hline
        Coefficients & Pauli Operators \\
        \hline
           - 0.678523 & IIII \\
           - 0.077605 & ZIII \\
           - 0.077605 & IZII \\
           + 0.134592 & ZZII \\
           - 0.077605 & IIZI \\
           + 0.222157 & ZIZI \\
           + 0.137722 & IZZI \\
           - 0.077605 & IIIZ \\
           + 0.137722 & ZIIZ \\
           + 0.222157 & IZIZ \\
           + 0.134592 & IIZZ \\
           - 0.291540 & XXII \\
           - 0.291540 & YYII \\
           + 0.001571 & XXZI \\
           + 0.001571 & YYZI \\
           + 0.001571 & XXIZ \\
           + 0.001571 & YYIZ \\
           - 0.291540 & IIXX \\
           + 0.001571 & ZIXX \\
           + 0.001571 & IZXX \\
           - 0.291540 & IIYY \\
           + 0.001571 & ZIYY \\
           + 0.001571 & IZYY \\
           + 0.003129 & XXXX \\
           + 0.003129 & YYXX \\
           + 0.003129 & XXYY \\
           + 0.003129 & YYYY \\
        \hline \hline
    \end{tabular}
    \caption{The Pauli sum representation of the H$_2$ electronic Hamiltonian at bond length of 0.5 \AA~ under STO-3G basis. 
    }
    \label{tab:h2_paulisum}
\end{table}

\bibliography{References}
\bibliographystyle{apsrev4-1}

\end{document}